%% file: paper_b_theory_spine.tex
\documentclass[11pt]{article}

\usepackage[margin=1in]{geometry}
\usepackage{amsmath,amssymb,amsthm}
\usepackage{bm}
\usepackage{booktabs}
\usepackage{enumitem}
\usepackage{hyperref}
\usepackage{natbib}
\usepackage[T1]{fontenc}
\usepackage[utf8]{inputenc}
\usepackage[sc,osf]{mathpazo}

\newcommand{\E}{\mathbb{E}}
\newcommand{\Pp}{\mathbb{P}}
\newcommand{\R}{\mathbb{R}}
\newcommand{\one}{\mathbf{1}}
\newcommand{\ind}{\mathbf{1}}
\newcommand{\sgn}{\operatorname{sgn}}
\newcommand{\logit}{\operatorname{logit}}
\newcommand{\Var}{\operatorname{Var}}
\newcommand{\Cov}{\operatorname{Cov}}
\newcommand{\BO}{\mathrm{BO}}

\newtheorem{theorem}{Theorem}[section]
\newtheorem{proposition}[theorem]{Proposition}
\newtheorem{corollary}[theorem]{Corollary}
\newtheorem{lemma}[theorem]{Lemma}
\theoremstyle{definition}
\newtheorem{assumption}[theorem]{Assumption}
\newtheorem{definition}[theorem]{Definition}
\theoremstyle{remark}
\newtheorem{remark}[theorem]{Remark}

\title{\bf Activity-Conditioned Residual Association\\
from Aggregated Relational Data}
\author{Yen-hsuan Tseng}
\date{}

\hypersetup{
  hidelinks,
  pdftitle={Activity-Conditioned Residual Association from Aggregated Relational Data},
  pdfauthor={Yen-hsuan Tseng},
  pdfkeywords={aggregated relational data, node activity, conditional inference, network diagnostics, sampled networks}
}

\begin{document}
\maketitle

\begin{abstract}
Aggregated relational data (ARD) record how many ties sampled respondents have
to prespecified groups without revealing individual dyads. We develop a
conditional test for prespecified cross-group concordance beyond an
additive-activity network model. When the groups form an exhaustive
partition, respondent degree is observed exactly. For arbitrary fixed activity
values under an independent-Bernoulli additive-logit null, conditioning two
respondents on equal degree yields a finite nonpositive sign restriction for
their two disjoint group-count differences. Analyst-randomized pair thinning
gives a conservative baseline. Under a stronger fixed-cell design
with group-specific smooth activity profiles and positive cross-role overlap,
an observable clipped full-pair statistic gives conservative one-network
inference with graph-independent sampled respondents. Finite and growing
matched-law constructions establish that the target contains bivariate
information absent from a collapsed binary count. The full-pair procedure is
uniformly consistent on a primitive relative-open neighborhood of a specified
rank-one alternative. Fixed-dimensional lattice local limits make the
conditioning and availability requirements explicit.
\end{abstract}


\section{Introduction}

ARD replace named dyads by respondent-level counts over trait groups. They
support personal-network-size estimation, population-size estimation, and
model-based inference about latent network structure
\citep{mccormick2010personal,breza2020ard,Breza2023}. These counts also provide
an observable basis for assessing whether node activity explains prespecified
patterns of group contact. We develop a conditional test of cross-group
concordance using the counts and degrees available from ARD.

Node activity is standard in network formation
\citep{graham2017econometric}; ARD models study correlated counts and residual
diagnostics \citep{laga2023correlated,laga2026diagnostics}; and recent work
examines both low-rank trait informativeness and flexible correlated ARD models
\citep{hayes2026spectral,seymour2026neural}. Building on these literatures, we
develop conditional inference for residual association from ARD. The method
combines exact activity conditioning and a finite sign restriction with one-network
calibration for sampled respondents. Under stronger smooth-design assumptions,
an observable full-pair studentizer improves information use. Matched-law
constructions establish the information retained by separate group queries.
An open-class power result and a separated-support counterexample link test
performance to the availability of degree-matched pairs.

A motivating use case is a closed organization with two named subgroups $G$
and $H$ inside a broader role $B=G\cup H$, and an exhaustive outside role $O$.
Before inspecting the network, an analyst may ask whether, among equally
connected respondents drawn from $B$ and $O$, the difference in reported ties
to $G$ tends to move in the same direction as the difference in reported ties
to $H$ more often than additive activity permits. A rejection provides evidence
of cross-group concordance beyond additive activity.

\paragraph{Design conditions.}
The finite sign restriction holds for arbitrary fixed activities under the
independent-Bernoulli additive-logit null.
The randomized baseline adds a repeated-cell activity array,
graph-independent positive-share respondent sampling, and analyst thinning.
The full-pair theorem uses fixed cell multiplicity, group-specific
increasing $C^{1,1}$ activity profiles, role-separated observed-degree
smoothing, and positive broad/outside activity overlap.  The local
interpretation uses a deterministic marked-profile assumption with a common
activity density.  The power theorem uses fixed-slot rank-one
$n^{-1/2}$ logit scaling and establishes uniform consistency on a primitive
relative-open neighborhood of the prespecified $c=6$ design.

The rest of the paper follows that logical order. Section~\ref{sec:experiment}
defines the observable experiment and the collision-normalized target.
Section~\ref{sec:exact} derives the exact finite null restriction.
Section~\ref{sec:baseline} gives the randomized baseline test, and
Section~\ref{sec:full-pair} gives the stronger full-pair procedure.
Section~\ref{sec:meaning} establishes finite nonredundancy and a local-response
interpretation. Section~\ref{sec:power} adds the growing matched-law
construction and primitive open-class power. Section~\ref{sec:finite} reports
finite-sample behavior under two null profiles, the declared path, two
respondent shares, and an unavailable design.

\section{Observable ARD Experiment and Target}
\label{sec:experiment}

Let $V_n=\{1,\ldots,n\}$ be a closed population with an unobserved undirected
network $A=(A_{ij})$, where $A_{ij}=A_{ji}$ and $A_{ii}=0$. The fixed,
disjoint groups $G_n,H_n,O_n$ form an
exhaustive partition of $V_n$, and $B_n=G_n\cup H_n$. For every sampled
respondent $i$, ARD reveal
\[
X_i=\sum_{k\in G_n}A_{ik},\qquad
Z_i=\sum_{k\in H_n}A_{ik},\qquad
Y_{iO}=\sum_{k\in O_n}A_{ik}.
\]
The sums may include the respondent's own group label; the zero diagonal makes
the self term identically zero.
Exhaustiveness makes the full closed-population degree observable:
\[
D_i=X_i+Z_i+Y_{iO}.
\]
Write $\Lambda(t)=(1+e^{-t})^{-1}$ for the logistic cdf.

For a pair of respondents, define
\[
h_{ij}
=
\ind\{D_i=D_j\}
\sgn(X_i-X_j)\sgn(Z_i-Z_j)
\]
and the prespecified broad/outside pair frame
\[
A_{ij}^{\BO}
=
\ind\{i\in B_n,j\in O_n\}
+\ind\{i\in O_n,j\in B_n\}.
\]
The frame is fixed before observing the network. Write
\[
K_n^{\BO}=\sum_{i<j}A_{ij}^{\BO}h_{ij},\qquad
M_n^{\BO}=\sum_{i<j}A_{ij}^{\BO}\ind\{D_i=D_j\}.
\]

\begin{definition}[Collision-normalized residual association]
Whenever $\E_P M_n^{\BO}>0$, define
\[
\tau_n(P)=\frac{\E_P K_n^{\BO}}{\E_P M_n^{\BO}}.
\]
The object summarizes the signed association between two disjoint group-count
differences among degree-matched broad/outside pairs.
\end{definition}

The sign records concordance: it is positive when the two group-count
differences point together and negative when they point apart. Under the
additive null, degree conditioning removes respondent activity from the
conditional leave-one-out row law. The pair frame is specified from role labels
before observing the network; degree matching uses the observed group counts.

Let $R_i$ indicate respondent inclusion. Respondents are sampled independently
of the graph under either common-probability Bernoulli sampling or fixed-size
simple random sampling with limiting share $\rho\in(0,1)$. The observable pair
weight is $W_{ij}=R_iR_jA_{ij}^{\BO}$.  Write
$S=(R_1,\ldots,R_n)$ for the realized respondent frame; all superscripts $S$
below refer to conditional one-network inference given that frame.

\section{Exact Activity Removal}
\label{sec:exact}

The null model allows arbitrary node activity:
\begin{equation}
\logit(p_{ij})=\mu+\alpha_i+\alpha_j,
\qquad i<j,
\label{eq:additive-null}
\end{equation}
with dyads conditionally independent given node attributes.

\begin{theorem}[Exact finite activity-null sign]
\label{thm:finite-sign}
Under \eqref{eq:additive-null}, fixed disjoint $G_n,H_n$, a common eligible
alter pool after deleting the focal pair, and any deterministic or
graph-independent respondent/pair frame,
\[
\E_0 h_{ij}\le0
\]
for every fixed pair. Consequently,
\[
\E_0\sum_{i<j}W_{ij}h_{ij}\le0,
\qquad
\tau_n(P_0)\le0
\]
whenever the target denominator is positive.
\end{theorem}

\paragraph{Proof sketch.}
Condition on the shared edge and equal total degree. The common respondent
odds factor cancels from each leave-one-out row, leaving two independent
weighted fixed-size samples from the same alter pool. Counts into disjoint
groups are negatively associated. The shared edge induces opposite-signed or
zero integer shifts, preserving the nonpositive sign product. Appendix
\ref{app:exact-sign} gives the conditional-row and shifted-sign lemmas and the
complete proof.

\begin{remark}[Range of validity]
The queried groups must be disjoint, and respondent selection must be
independent of the graph. Overlapping groups can give the shared-edge shifts
the same sign, while graph-selected respondents can reweight network events.
\end{remark}

\section{Randomized Baseline Calibration}
\label{sec:baseline}

\begin{assumption}[Baseline regular array]
\label{ass:baseline}
Let $n=40m_n$. The population has $m_n$ activity cells of size 40, with cell
centers given by deterministic midpoint quantiles of a fixed distribution
$F$ whose $C^1$ density $f$ on a compact interval $I$ is bounded above and away
from zero. Every activity cell
contains the exact group ratio
$|G|:|H|:|O|=10:10:20$. Respondents follow one of the two graph-independent
positive-share designs above.
\end{assumption}

Given the observed network and respondent frame, independently draw analyst
masks for eligible pairs,
\[
J_{ij}\sim\operatorname{Bernoulli}(\pi_n),
\]
and define
\[
S_{2,R,n}^{\BO}=\sum_{i<j}W_{ij}h_{ij}^2,
\]
\[
Z_n=
\frac{\sum_{i<j}J_{ij}W_{ij}h_{ij}}
{\{\pi_n(1-\pi_n)S_{2,R,n}^{\BO}\}^{1/2}},
\]
with $Z_n=0$ when the denominator vanishes.

\begin{theorem}[Conservative asymptotic one-network baseline test]
\label{thm:thinning-null}
Under Assumption~\ref{ass:baseline}, the additive null, either declared
respondent design, and $\pi_n=n^{-1/2}$,
\[
\limsup_{n\to\infty}
\Pp_0\{Z_n>z_{1-\alpha}\}\le\alpha.
\]
More generally, for $\pi_n=n^{-\beta}$ with $0<\beta<1$, the proved
calibration envelope is
\[
\alpha+O\{n^{-\beta/2}+n^{-(1-\beta)/2}+n^{-1/2}\}.
\]
The choice $\beta=1/2$ balances the two leading terms in this calibration
envelope.
\end{theorem}

\paragraph{Proof outline.}
The proof combines Theorem~\ref{thm:finite-sign}, a conditional
Berry--Esseen bound for the analyst mask, order-$n$ collision mass, order-$n$
network variance, and graph-independent respondent-frame transfer. The
randomization attenuates the unobserved network center and can be
conservative at practical sizes. Appendix~\ref{app:baseline-proof} states the
needed local-mass lemmas and proves both the null limit and the displayed
polynomial calibration envelope.

The analyst mask is an asymptotic calibration device used jointly with the
network law and graph-independent respondent design.

\section{Full-Pair One-Network Inference}
\label{sec:full-pair}

The thinning theorem isolates a simple calibration mechanism, but discards
most degree-matched pairs. Stronger design conditions permit use of the
entire natural pair frame with an observable correction for shared-row
dependence.

\begin{assumption}[Regular full-pair profile]
\label{ass:full-pair-profile}
Let $n=Lm_n$, where $L$ and the positive integers $c_G,c_H,c_O$ are fixed and
sum to $L$.  Every activity cell contains exactly these role counts. At
$u_{k,n}=(k-1/2)/m_n$, a group-$g$ node has activity
\[
\alpha_{g,k}=m_g(u_{k,n}),\qquad g\in\{G,H,O\},
\]
where the three fixed maps belong to a common $C^{1,1}$ ball,
\[
0<c_m\le m_g'\le C_m,\qquad
\operatorname{Lip}(m_g')\le L_m,
\]
and all additive-logit predictors lie in one compact interval. Respondents
follow either declared graph-independent design with limiting share
$\rho\in(0,1)$.

Put $\pi_g=c_g/L$, $I_g=m_g([0,1])$, and, with zero extension outside $I_g$,
\[
\lambda_g(a)
=\frac{\pi_g}{m_g'\{m_g^{-1}(a)\}}\ind\{a\in I_g\}.
\]
The broad/outside overlap satisfies
\[
\Omega_{BO}
=\int\{\lambda_G(a)+\lambda_H(a)\}\lambda_O(a)\,da>0.
\]
Within sampled $B$ and $O$ respondents separately, order nodes by observed
full degree using a deterministic prespecified tie-break and form consecutive
balanced blocks of nominal sampled-role size to the $2/3$ power, merging the
terminal remainder. Thus their block sizes $\ell_n$ satisfy
\[
\ell_n\to\infty,\qquad
\ell_n/n\to0,\qquad
\frac{\sqrt{n\log n}}{\ell_n}\to0.
\]
\end{assumption}

For a realized respondent frame, let
\[
\mathcal E_n^S
=\{\{i,j\}:i<j,\ R_iR_jA_{ij}^{\BO}=1\}
\]
be the sampled natural pair set and define
\[
U_n^S=\sum_{e\in\mathcal E_n^S}h_e,\qquad
s_i^S=\sum_{e\in\mathcal E_n^S:e\ni i}h_e,\qquad
A_n^S=\sum_{e\in\mathcal E_n^S}h_e^2,
\]
\[
T_n^S=\sum_i(s_i^S)^2-2A_n^S.
\]
Let $\mathcal C_i^S$ be the observed-degree block containing $i$.  The
leave-one-out row smoother is
\[
\widehat b_i^S
=\frac{1}{|\mathcal C_i^S|-1}
\sum_{j\in\mathcal C_i^S\setminus\{i\}}s_j^S,
\]
with $\widehat b_i^S=0$ off the sampled frame.

\begin{definition}[Observable clipped full-pair procedure]
\label{def:full-pair-procedure}
Define
\[
\widehat V_n^S
=A_n^S+T_n^S-\sum_i(\widehat b_i^S)^2.
\]
With the deterministic scale window
\[
a_n=\{\log(n+e)\}^{-2},\qquad b_n=\log(n+e),
\]
set
\[
\widehat V_{n,\dagger}^S
=\{\widehat V_n^S\vee na_n\}\wedge nb_n,\qquad
\mathcal T_{n,\dagger}^S
=\frac{U_n^S}{(\widehat V_{n,\dagger}^S)^{1/2}}.
\]
The full-pair test rejects at one-sided level $\alpha$ when
$\mathcal T_{n,\dagger}^S>z_{1-\alpha}$.
\end{definition}

The role-separated blocks, degree ordering, tie-break, and clipping window
are prespecified parts of the procedure.

\begin{theorem}[Conservative full-pair one-network inference]
\label{thm:full-pair-null}
\leavevmode\par\noindent
Under Assumption~\ref{ass:full-pair-profile} and the additive-logit null with
independent dyads, there are graph-independent
respondent-frame events $\mathcal G_n$, with
$\Pp_S(\mathcal G_n)\to1$, and constants $0<c<C<\infty$ such that, uniformly
over $S\in\mathcal G_n$,
\[
cn\le \sigma_{n,S}^2
:=\Var_0(U_n^S\mid S)\le Cn,
\]
\[
\frac{\widehat V_n^S}{\sigma_{n,S}^2}
\longrightarrow_{P_0(\cdot\mid S)}1,
\]
and
\[
\frac{U_n^S-\E_0(U_n^S\mid S)}{\sigma_{n,S}}
\Rightarrow N(0,1).
\]
Since $\E_0(U_n^S\mid S)\le0$ exactly, for every fixed
$\alpha\in(0,1/2)$,
\[
\limsup_n\sup_{S\in\mathcal G_n}
\Pp_0\{\mathcal T_{n,\dagger}^S>z_{1-\alpha}\mid S\}
\le\alpha.
\]
The same conservative bound holds after averaging over either declared
respondent design.
\end{theorem}

\paragraph{Proof outline.}
Fixed cell multiplicity yields exact group-specific activity intensities and
a continuous local collision kernel away from finitely many support endpoints.
Absolute same-row and four-distinct bounds control block noise and the sampled
four-distinct remainder. Observed degree ranks consistently order the smooth
row-mean curve, so the leave-one-out block smoother removes the first-order
row projection. Positive broad/outside overlap supplies order $n^{3/2}$
close sampled pairs and hence an order-$n$ variance lower bound. Fixed-order
cumulants then give the centered conditional normal limit; the exact
nonpositive null mean turns it into a feasible conservative test. Appendix
\ref{app:full-pair-proof} gives the complete proof.

\begin{remark}[Activity overlap]
The test calibrates against zero and estimates its variance from observed
pair scores and smoothed row sums. Positive $\Omega_{BO}$ ensures sufficient
overlap between the two respondent roles for the variance lower bound.
If $\Omega_{BO}=0$, the nuisance controls can remain valid while root-$n$
full-pair information fails.
\end{remark}

\section{Scientific Meaning Beyond a Collapsed Binary Count}
\label{sec:meaning}

If an analyst queries only the broad role $B=G\cup H$, the observable core
count is $C_i=X_i+Z_i$ and the allocation of ties between $G$ and $H$ is lost.
The next results compare this natural coarser ARD design with separate queries
for $G$ and $H$, holding the joint law of the shared edge, degrees, and
collapsed counts fixed.
The comparison identifies the bivariate information retained by separate
queries.

\begin{proposition}[Finite bivariate nonredundancy]
\label{prop:nonredundancy}
There exist two finite logistic-Gram probability surfaces with the same joint
law of the shared edge, both respondents' total degrees, and both respondents'
collapsed $G\cup H$ counts, but with pair-sign targets of opposite sign.
Hence the target is not determined by the corresponding binary-core/degree
experiment.
\end{proposition}

\paragraph{Proof sketch.}
Keep each focal respondent's multiset of core-edge probabilities fixed while
reassigning those probabilities between $G$ and $H$. The collapsed core and
degree law is then unchanged, whereas exact finite enumeration reverses the
pair-sign numerator. Appendix~\ref{app:finite-witness} gives the rational
construction and a rank-two logistic-Gram embedding.

\subsection{Local interpretation after activity profiling}

The finite witness establishes bivariate information beyond the collapsed
count. A separate question is how a small residual interaction
moves the target once total-degree information has removed the local activity
direction.  Consider the rank-one sequence
\[
\logit(p_{ij,c})
=\mu+\alpha_i+\alpha_j+\frac{c}{\sqrt n}x_ix_j.
\]
Appendix~\ref{app:local-response} states the deterministic marked-profile
assumption and defines $b_G,b_H$ as the standardized residual loading shifts
after total-count projection, $\varrho$ as the resulting conditional Gaussian
correlation, and $\chi_{BO}$ as the focal loading contrast. Let $R(c)$ denote
the leading Gaussian expectation of the pair-sign kernel conditional on a
degree collision in the natural $B/O$ frame. Central symmetry makes this
response even in $c$.

\begin{proposition}[Activity-profiled local response]
\label{prop:curvature}
Under Assumption~\ref{ass:marked-profile}, at a regular common-support
collision midpoint, the
collision-conditioned Gaussian response satisfies
\[
R(c)=R(0)+\frac12\kappa c^2+o(c^2),
\]
where $\kappa$ is an explicit quadratic functional of the projected loading
masses and the degree-conditioned Gaussian covariance:
\[
\kappa
=\chi_{BO}(\bar a)\frac{2}{\pi\sqrt{1-\varrho^2}}
\{2b_Gb_H-\varrho(b_G^2+b_H^2)\}.
\]
For the declared broad and split loading profiles, the curvature is strictly
positive and negative, respectively, throughout the compact activity range,
and the two profiles therefore belong to open response classes.
\end{proposition}

\begin{remark}[Local mean information]
\label{rem:curvature-information}
At the same collision midpoint, put $b=(b_G,b_H)^\top$ and
\[
R_{\varrho}=\begin{pmatrix}1&\varrho\\\varrho&1\end{pmatrix},
\qquad I_{\mathrm{mean}}=b^\top R_{\varrho}^{-1}b.
\]
This is the Fisher information for $t$ in the Gaussian mean experiment
$N(tb,R_{\varrho})$, with the degree-conditioned covariance fixed. The
curvature satisfies the sharp bound
\[
|\kappa|
\le\frac{2\chi_{BO}(\bar a)}{\pi}
\sqrt{1-\varrho^2}\,I_{\mathrm{mean}}.
\]
Equality is attained on the sum axis $b_G=b_H$ and the difference axis
$b_G=-b_H$, with nonnegative and nonpositive curvature, respectively.
Thus local mean information bounds the response magnitude,
whereas its allocation between the sum and difference modes determines the
sign. Appendix~\ref{app:local-response} gives the modal decomposition.
\end{remark}

\paragraph{Proof sketch.}
The Gaussian row-difference limit is conditioned on total count by a Schur
complement. The same projection removes any information-weighted loading
profile common to all three groups. Differentiating the resulting bivariate
normal sign probability twice at the balanced origin gives the quadratic
form displayed in Appendix~\ref{app:local-response}.

\section{Primitive Open-Class Power}
\label{sec:power}

Fix $L=40$ cell slots with group counts
$|G|:|H|:|O|=10:10:20$ and fixed slot labels $g(r)$. A primitive parameter is
\[
\vartheta=(m_G,m_H,m_O,x,c,\mu)
\]
in the relative space
\[
(C^1[0,1])^3
\times\{x\in\R^{40}:\one^\top x=0\}
\times\R^2,
\]
intersected with the common $C^{1,1}$, monotonicity, bounded-loading, compact
$c,\mu$, and compact-predictor restrictions. The activity-profile part of this
admissible class is the one in Assumption~\ref{ass:full-pair-profile}; the
loading and $c$ bounds are additional alternative-class restrictions. For slot
$r$ in cell $k$, put
\[
\alpha_{r,k}=m_{g(r)}\{(k-1/2)/m_n\},
\]
and define
\[
\logit(p_{ij,\vartheta})
=\mu+\alpha_i+\alpha_j
+\frac{c}{\sqrt n}x_{r(i)}x_{r(j)}.
\]
The centering condition fixes an additive loading gauge; the scale
representation $(c,x)$ remains nonunique.

The canonical parameter $\vartheta_0$ has the common map
$m_0=F^{-1}$ from Assumption~\ref{ass:baseline}, $\mu=0$, $c=6$, binary
loadings, and prespecified counts
\[
(G_+,G_-,H_+,H_-,O_+,O_-)=(7,3,7,3,6,14).
\]
For $\varepsilon>0$, let $\Theta_\varepsilon$ be the relative-open primitive
ball
\[
\max_g\|m_g-m_0\|_{C^1}<\varepsilon,\qquad
\|x-x^0\|_\infty<\varepsilon,\qquad
|c-6|<\varepsilon,\qquad|\mu|<\varepsilon,
\]
within the admissible class. It permits unequal smooth group activity profiles
and nonbinary centered loadings near $\vartheta_0$. The operational frame
remains every observable $B/O$ pair; latent loadings are never used to select
pairs.

For a standard bivariate normal vector $(N_1,N_2)$ with correlation $\varrho$,
write
\[
g_\varrho(\theta)
=\E\{\sgn(N_1+\theta)\sgn(N_2+\theta)\}.
\]

\begin{assumption}[Fixed-path response margin]
\label{ass:response-margin}
Let
\[
v(a)=\int\Lambda'(\mu+a+t)\,dF(t),
\qquad
v_*=\inf_{a\in I}v(a).
\]
At $c=6$, assume
\[
0.58g_{-1/3}\{2.4\sqrt{2v_*/3}\}
+0.42g_{-1/3}(0)>0.
\]
\end{assumption}

\begin{remark}
The margin is nonvacuous.  For example, $\mu=0$ and uniform $F$ on
$[-0.6,0.6]$ give $v_*\ge0.2236888$; the corresponding positivity threshold
for $c$ is approximately $4.606<6$.
\end{remark}

\subsection{Growing nonredundancy on the declared path}

\begin{proposition}[Growing nonredundancy]
\label{prop:growing-nonredundancy}
Under Assumptions~\ref{ass:baseline} and~\ref{ass:response-margin}, there are
two growing logistic-Gram arrays with rank-one positive-semidefinite residuals.
For matched $B/O$ focal pairs, at every $n$ the full collapsed law
of
\[
(A_{ij},C_i,C_j,D_i,D_j),
\qquad C_r=X_r+Z_r,
\]
is identical across the two arrays. On the declared $c=6$ path, their
degree-conditioned $G/H$ pair-sign responses have opposite limiting signs.
Both arrays are nonadditive; their matched collapsed laws establish the
additional information carried by the bivariate experiment.
\end{proposition}

\paragraph{Construction.}
Repeat 40-node activity cells over an increasingly fine midpoint-quantile
grid. In each cell, the broad allocation
$(G_+,G_-,H_+,H_-,O_+,O_-)=(7,3,7,3,6,14)$ and the split allocation
$(10,0,4,6,6,14)$ contain the same multiset of activity/loading pairs in
$B=G\cup H$ and the same multiset in $O$. For matched focal identities,
activities, and loadings, the Bernoulli probability-pair generating functions
therefore agree separately on $B$ and $O$, as does the shared-edge law. The
group allocation retained by the bivariate ARD counts produces positive broad
and negative split response. Appendix~\ref{app:growing-witness} records the
law-matching argument. The broad array anchors the power theorem below, which
uses the observable natural pair frame.

For slot types $r\in B$ and $s\in O$, activity anchor $a$, and scaled activity
gap $t$, let $H_\vartheta(r,s,a,t)$ and
$J_\vartheta(r,s,a,t)$ denote the limiting kernels for
$\sqrt n\,\E_\vartheta h_{ij}$ and
$\sqrt n\,\E_\vartheta h_{ij}^2$.  Appendix
\ref{app:full-pair-power} constructs them by information-weighted recentering
of the degree constraint.  With slot intensity
\[
\lambda_{r,\vartheta}(a)
=\frac{\ind\{a\in m_{g(r)}([0,1])\}}
{L\,m'_{g(r)}\{m_{g(r)}^{-1}(a)\}},
\]
define
\[
\mathcal M(\vartheta)
=\sum_{r\in B}\sum_{s\in O}
\int\!\!\int
\lambda_{r,\vartheta}(a)\lambda_{s,\vartheta}(a)
H_\vartheta(r,s,a,t)\,dt\,da,
\]
\[
\mathcal C(\vartheta)
=\sum_{r\in B}\sum_{s\in O}
\int\!\!\int
\lambda_{r,\vartheta}(a)\lambda_{s,\vartheta}(a)
J_\vartheta(r,s,a,t)\,dt\,da.
\]
These population response and collision constants enter the power proof. The
operational statistic uses only observed counts, degrees, role labels, and
respondent indicators.

Write $\mathsf S_n=(R_1,\ldots,R_n)$ for the random respondent frame, reserving
$S$ for a fixed realization as in Section~\ref{sec:experiment}.

\begin{theorem}[Open-class full-pair consistency]
\label{thm:full-pair-open-power}
Under Assumption~\ref{ass:response-margin}, there exist
$\varepsilon_0>0$ and $\gamma_0>0$ such that
\[
\inf_{\vartheta\in\Theta_{\varepsilon_0}}
\frac{\mathcal M(\vartheta)}{\mathcal C(\vartheta)}
\ge\gamma_0>0
\]
and, under either joint respondent-design/network law,
\[
\inf_{\vartheta\in\Theta_{\varepsilon_0}}
\Pp_\vartheta\{
\mathcal T_{n,\dagger}^{\mathsf S_n}>z_{1-\alpha}\}
\longrightarrow1.
\]
\end{theorem}

\paragraph{Proof outline.}
At $\vartheta_0$, the observable $B/O$ frame decomposes into a mixture of
opposite-sign and same-sign focal pairs with weights $0.58$ and $0.42$, and
Assumption~\ref{ass:response-margin} makes its response strictly positive.
For nearby group-specific activity profiles, an information-weighted
$n^{-1/2}$ correction recenters the degree constraint. The resulting marked
local-limit kernels are jointly continuous in the activity maps, centered
loadings, $c$, and $\mu$; thus the canonical response has a positive margin on
a primitive relative-open neighborhood. Uniformly on that neighborhood, the
full-pair numerator has a positive order-$n$ mean and $O(n)$ variance. The
deterministic upper clipping cap converts these moment bounds into consistency.
Appendix
\ref{app:full-pair-power} gives the details.

\begin{corollary}[Canonical thinning consistency]
\label{cor:thinning-power}
Under Assumptions~\ref{ass:baseline} and~\ref{ass:response-margin}, the
canonical mixed-sign design at $c=6$, either respondent design, and
$\pi_n=n^{-1/2}$, there exists $\gamma_6>0$ such that
\[
\liminf_n\tau_n(P_{n,6})\ge\gamma_6,\qquad
\Pp_6\{Z_n>z_{1-\alpha}\}\longrightarrow1.
\]
\end{corollary}

The theorem gives uniform consistency over the fixed-slot rank-one
neighborhood under the joint respondent-design/network law. The split design
has negative response, while separated activity supports can make collision
mass collapse. These two examples distinguish the direction of the residual
association from the availability of degree-matched pairs.

\section{Finite-Sample Evidence}
\label{sec:finite}

This section reports a prespecified finite-sample study of the full-pair
procedure and its thinning baseline. The design crosses common and group-specific
smooth activity profiles, Bernoulli and fixed-size-SRS respondent sampling,
respondent shares $0.35$ and $0.60$, and
$n\in\{80,160,320,640\}$.  The two null designs set $c=0$; the two P1 designs
use the prespecified broad rank-one loading profile at $c=6$. Every cell uses
the same 300-replication design.

\input{pair_sign_minimal_t5_finite_evidence_table.tex}

The full-pair rule is conservative in both null profiles. Under P1 it
is materially more sensitive than thinning, although finite rejection remains
moderate at the smaller respondent share and under the group-specific stress.
At the prespecified largest design, all eight paired full-minus-thinning
intervals have positive lower endpoints. The intervals describe paired
finite-sample differences and are unadjusted for multiple comparisons.

The separated-support design isolates the role of activity overlap. It
retains the broad loading proportions and smooth activity profiles while
placing the two respondent roles on supports separated by a fixed gap.

\input{pair_sign_minimal_t5_availability_table.tex}

As $n$ grows, nonzero pair density collapses, the deterministic lower scale
floor activates, and rejection disappears. The finite experiment therefore
illustrates the loss of usable degree-matched pairs when the two roles have
disjoint activity supports.

\section{Related Work}

Several neighboring literatures answer different parts of the problem.
Model-based ARD methods recover network-model parameters or probability
structure \citep{breza2020ard,Breza2023,alidaee2020recovering}, while recent
work studies low-rank trait informativeness and flexible likelihood-free ARD
estimation \citep{hayes2026spectral,seymour2026neural}. Our question is
complementary: which target-specific implication of an additive-activity null
survives ARD compression?

ARD model checking is also established. \citet{laga2026diagnostics} provide a
general diagnostic workflow that includes residual correlation checks, and
\citet{lubold2025spectral} develop spectral goodness-of-fit tests for complete
and partial network data, including ARD. \citet{laga2023correlated} model
cross-group dependence in network-scale-up counts through correlated
respondent-level random effects. We condition on
an exactly observed degree and test one prespecified bivariate sign restriction
without fitting a latent residual or correlation structure.

Finally, conditioning on degree to remove node heterogeneity is classical when
the dyads themselves are observed. For example, \citet{karwa2024montecarlo}
construct finite-sample conditional goodness-of-fit tests for degree-corrected
blockmodels with known block assignments, using sufficient statistics and
Markov-basis sampling. The present
experiment carries degree conditioning through ARD group counts rather than an
observed network fiber. The procedure combines an exact null
restriction, sampled-respondent one-network calibration, and an explicit
collision requirement, with a separated-support design showing when the
procedure becomes unavailable.
Pair thinning is a randomization device, while uniform local expansions for
nonidentical Bernoulli sums and general multivariate lattice local-limit theory
are classical
\citep{arratia2005local,gamkrelidze2015local}; the appendix verifies the
specific maximal-lattice and marked-conditioning conditions needed for the
target-specific ARD restriction and its nested sampled-respondent calibration
procedures, including observable full-pair studentization under the
regular-profile design.

\section{Discussion}

Conditioning on equal degree yields an exact nonpositive sign restriction
under the additive-activity null. With an exhaustive group partition and a
common alter pool, this restriction turns a prespecified cross-group question
into an observable ARD diagnostic. Randomized thinning provides one-network
calibration under the repeated-cell design. Under the regular full-pair
profile, observed-degree smoothing accounts for shared-row dependence and
allows inference using all eligible pairs. The resulting test is uniformly
consistent on a primitive open neighborhood of the declared rank-one design.

The analysis links a prespecified cross-group contrast, a null restriction
preserved by aggregation, and calibration for network-induced dependence.
The overlap condition identifies when the survey supplies enough
degree-matched respondents for inference.

A rejection provides evidence of residual concordance beyond additive
activity. The matched-law constructions show why separate group queries matter:
the same joint law of collapsed counts and degrees can accompany opposite
pair-sign responses. Failure to reject remains compatible with
nonadditive structure, including alternatives with negative response and
designs with insufficient degree collisions. Interpreting the underlying
social mechanism requires information beyond this prespecified contrast.

\section*{Code and Data Availability}

All finite evidence in this paper is synthetic. Replication materials include
the simulation code, machine-readable outputs, table-generation code, and a
reproducibility manifest sufficient to regenerate the reported tables. These
materials are available from the author upon request.

\appendix

\section{Proof of the Exact Finite Sign Restriction}
\label{app:exact-sign}

Fix two focal respondents $i$ and $j$, condition on their shared edge
$A_{ij}=e$, and delete the focal nodes from the eligible alter pool.

\begin{lemma}[Conditional row invariance]
\label{lem:row-invariance}
Under the additive-logit null, conditional on an external row size $d$, a
neighbor set $S$ in the common alter pool has law
\[
 \Pp_i(S\mid |S|=d)
 \ \propto
 \prod_{k\in S}e^{\alpha_k}.
\]
This law does not depend on $\mu+\alpha_i$. Hence, conditional on
$A_{ij}=e$ and $D_i=D_j=d+e$, the two external rows are independent and
identically distributed weighted fixed-size samples of common size $d$.
\end{lemma}

\begin{proof}
For one external row realization $S$, independence of the incident dyads gives
\[
 \Pp_i(S)
 \ \propto
 \prod_{k\in S}\frac{p_{ik}}{1-p_{ik}}
 =\exp\{|S|(\mu+\alpha_i)\}\prod_{k\in S}e^{\alpha_k}.
\]
The first factor is constant on $|S|=d$ and cancels after conditioning.  The
same conditional law holds for $j$, and the two external edge sets are
disjoint once the shared edge is fixed.
\end{proof}

The generating polynomial of the product-weighted fixed-size subset law in
Lemma~\ref{lem:row-invariance} is a weighted elementary symmetric polynomial.
It is real stable, so the law is strongly Rayleigh and hence negatively
associated \citep{borcea2009negative}.  Counts over disjoint subsets inherit
that property.

\begin{lemma}[Shifted sign under negative association]
\label{lem:shifted-sign}
Let $(X,Z)$ and $(X',Z')$ be iid integer-valued vectors.  If $(X,Z)$ is
negatively associated, then, for $a,b\in\{-1,0,1\}$ with $ab\le0$,
\[
 \E\{\sgn(X-X'+a)\sgn(Z-Z'+b)\}\le0.
\]
\end{lemma}

\begin{proof}
Conditional on $(X',Z')$, both sign maps are increasing in $(X,Z)$.
Negative association therefore bounds the conditional product expectation by
$f_a(X')g_b(Z')$, where
\[
 f_a(x')=\E\{\sgn(X-x'+a)\},\qquad
 g_b(z')=\E\{\sgn(Z-z'+b)\}.
\]
Both functions are decreasing.  Applying negative association to
$-f_a(X')$ and $-g_b(Z')$ gives
\[
 \E\{f_a(X')g_b(Z')\}\le \E f_a(X')\,\E g_b(Z').
\]
If $\Delta_X=X-X'$, symmetry of the iid difference implies
\[
 \E\sgn(\Delta_X+a)=a c_X,
 \qquad
 c_X=\Pp(\Delta_X=0)+\Pp(\Delta_X=1)\ge0,
\]
with the same formula for $Z$.  The final upper bound is $ab c_Xc_Z\le0$.
\end{proof}

\begin{proof}[Proof of Theorem~\ref{thm:finite-sign}]
Write $(X_i^-,Z_i^-)$ and $(X_j^-,Z_j^-)$ for counts after deleting the two
focal nodes.  On $A_{ij}=e$, the observed differences equal
\[
 X_i-X_j=X_i^--X_j^-+a,
 \qquad
 Z_i-Z_j=Z_i^--Z_j^-+b,
\]
where
\[
 a=e\{\ind(j\in G_n)-\ind(i\in G_n)\},\qquad
 b=e\{\ind(j\in H_n)-\ind(i\in H_n)\}.
\]
The disjointness of $G_n$ and $H_n$ gives $ab\le0$.  Conditional on
$A_{ij}=e$ and $D_i=D_j=d+e$,
Lemma~\ref{lem:row-invariance} makes the two external rows iid weighted
fixed-size samples.  Their $G_n$ and $H_n$ counts are negatively associated,
so Lemma~\ref{lem:shifted-sign} gives the desired conditional nonpositive
expectation.  Summing over $d$ and the shared-edge state proves
$\E_0h_{ij}\le0$.  A deterministic or graph-independent frame only
multiplies each pair contribution by a nonnegative factor independent of the
graph, which proves the sum and target statements.
\end{proof}

\section{Local Mass, Network Center, and Baseline Calibration}
\label{app:baseline-proof}

This section records the technical claims used by the baseline theorem for
the exact repeated-cell experiment in Assumption~\ref{ass:baseline}.  All
bounds are uniform for $c$ in a fixed compact set and for bounded lattice
offsets. Theorem~\ref{thm:thinning-null} uses these bounds at $c=0$, and
Corollary~\ref{cor:thinning-power} uses them at $c=6$ under the cellwise
mixed-sign loading balance from Section~\ref{sec:power}.

\subsection{Uniform lattice and marked-conditioning tools}

Let $\omega$ index the finite phase, focal-label, conditioned-edge, bounded
offset, and compact local-coefficient choices in the declared arrays.

\begin{lemma}[Deterministic marked-profile summation]
\label{lem:marked-profile-summation}
Under Assumption~\ref{ass:marked-profile}, for every bounded-variation
function $q$ and $r\in\{0,1,2\}$,
\[
\frac1n\sum_{i:g_i=A}x_i^r q(\alpha_i)
\longrightarrow
\int_Iq(a)\nu_{A,r}(a)\,da.
\]
At every regular interior point $a$ and fixed $s<t$,
\[
\frac1{\sqrt n}\sum_{i:g_i=A}x_i^r
\ind\left\{a+\frac{s}{\sqrt n}<\alpha_i
\le a+\frac{t}{\sqrt n}\right\}
\longrightarrow(t-s)\nu_{A,r}(a).
\]
Moreover, for every fixed $M$ and $c>0$,
\[
\sup_i\sum_j
\{1+|\sqrt n(\alpha_i-\alpha_j)|\}^M
e^{-cn(\alpha_i-\alpha_j)^2}=O(\sqrt n).
\]
For disjoint roles $A,C$, marked orders $r,s\in\{0,1,2\}$, and a bounded-
variation kernel $\phi$ with compact support or Gaussian tail,
\[
\frac1{n^{3/2}}
\sum_{i:g_i=A}\sum_{j:g_j=C}
x_i^r x_j^s\phi\{\sqrt n(\alpha_j-\alpha_i)\}
\longrightarrow
\left\{\int_{\mathbb R}\phi(u)\,du\right\}
\int_I\nu_{A,r}(a)\nu_{C,s}(a)\,da.
\]
The repeated midpoint-quantile arrays satisfy these statements with the
intensities displayed after Assumption~\ref{ass:marked-profile}.
\end{lemma}

\begin{proof}
Let
\[
R_{A,r,n}(a)
=F_{A,r,n}(a)-n\int_{I_-}^{a}\nu_{A,r}(t)\,dt.
\]
Stieltjes integration by parts bounds the discrepancy in the first display by
$n^{-1}\|R_{A,r,n}\|_\infty$ times the variation and endpoint size of $q$,
which is $o(n^{-1/2})$. Taking differences of $F_{A,r,n}$ over a
root-$n$ interval gives the second display at continuity points. For the last
display, split the real line into unit annuli on the $\sqrt n$ scale.
Assumption~\ref{ass:marked-profile}'s packing bound supplies $O(\sqrt n)$
nodes in every such annulus, while the polynomially weighted Gaussian factors
are summable over the annulus index. For the cross-role sum, replace the inner marked empirical measure
by $n\nu_{C,s}(y)dy$ using the same Stieltjes bound. The resulting outer
integrand has variation $O(\sqrt n)$, so a second replacement contributes
$o(n^{3/2})$; change variables and dominated convergence give the display.
Midpoint-quantile counting has uniformly bounded cumulative discrepancy, and
the density bounds give the required packing envelope.
\end{proof}

\begin{lemma}[Uniform unimodular lattice bounds]
\label{lem:uniform-lattice}
For $d\in\{1,2\}$, let
\[
S_{n,\omega}=\sum_{\ell=1}^{N_{n,\omega}}X_{n\ell,\omega},
\qquad X_{n\ell,\omega}\in\mathbb Z^d,
\]
be a sum of independent uniformly bounded increments with
$N_{n,\omega}\asymp n$. Suppose there are disjoint index blocks
$J_{1,n},\ldots,J_{d,n}$, each of order $n$, and integer vectors
$v_1,\ldots,v_d$ such that every increment in $J_{r,n}$ assigns probability
at least $\epsilon>0$ to two support atoms differing by $v_r$, and
$|\det(v_1,\ldots,v_d)|=1$. Then, uniformly in $\omega$,
\[
cnI_d\preceq\Sigma_{n,\omega}:=\Var(S_{n,\omega})
\preceq CnI_d,
\]
and, for $t\in[-\pi,\pi]^d$,
\[
\left|\E e^{it^\top(S_{n,\omega}-\E S_{n,\omega})}\right|
\le
\exp\{-cn\operatorname{dist}(t,2\pi\mathbb Z^d)^2\}.
\]
The maximal lattice is therefore $\mathbb Z^d$. If
$\phi_{\Sigma}$ denotes the centered Gaussian density with covariance
$\Sigma$, then for every fixed $M<\infty$ and
$\|z-\E S_{n,\omega}\|\le M\sqrt n$,
\[
\Pp(S_{n,\omega}=z)
=\phi_{\Sigma_{n,\omega}}(z-\E S_{n,\omega})
+O\{n^{-(d+1)/2}\}.
\]
Consequently,
\[
\sup_z\Pp(S_{n,\omega}=z)\le Cn^{-d/2},
\]
the mass is bounded above and below by constant multiples of $n^{-d/2}$
on every bounded standardized region, and, for every fixed lattice direction
$v$,
\[
\sup_z|\Pp(S_{n,\omega}=z+v)-\Pp(S_{n,\omega}=z)|
\le Cn^{-(d+1)/2}.
\]

These conclusions survive the bounded exponential tilts used below. More
precisely, if the change-of-measure factor at a target $z$ is bounded above by
$\exp(-cnQ+C)$ and the tilt has norm at most
$C\{Q^{1/2}+n^{-1/2}\}$, then the point-mass and first-difference bounds gain
the common envelope $\exp(-c'nQ)$. If, in addition, the standardized tilted
target is bounded and the reverse change-of-measure bound
\[
K_n(\lambda)-\lambda^\top z\ge-C_M
\]
holds whenever $nQ\le M$, then a matching local lower bound holds on that
region.
\end{lemma}

\begin{proof}
For one anchor increment with atoms $a$ and $a+v_r$, an independent-copy
identity gives
\[
|\E e^{it^\top X}|^2
\le1-2\epsilon^2\{1-\cos(t^\top v_r)\}.
\]
Multiplication over the disjoint anchor blocks yields
\[
|\E e^{it^\top S_n}|
\le
\exp\left[-cn\sum_{r=1}^d
\sin^2(t^\top v_r/2)\right].
\]
Unimodularity makes $2\pi\mathbb Z^d$ the only common zero set and also gives
the covariance lower bound; bounded increments give the upper bound.

Fourier inversion on the maximal lattice, together with the uniform expansion
\[
\log\E e^{it^\top(S_n-\E S_n)}
=-\tfrac12t^\top\Sigma_nt+O(n\|t\|^3),
\]
gives the local Gaussian approximation on the major arc. The characteristic-
function bound makes the minor arc exponentially small. Inserting
$e^{-it^\top v}-1$ into the inversion integral gives the extra
$n^{-1/2}$ factor for a first lattice difference. Finally, exponential
tilting gives
\[
\Pp(S_n=z)=e^{K_n(\lambda)-\lambda^\top z}
\Pp_\lambda(S_n=z).
\]
Applying the local bounds under the tilted law, and using
$(Q^{1/2}+n^{-1/2})e^{-cnQ}\le Cn^{-1/2}e^{-c'nQ}$ for adjacent targets,
proves the upper envelope and first-difference bound. The central Gaussian
lower mass together with the stated reverse change-of-measure bound proves the
local lower bound.
\end{proof}

\begin{lemma}[Marked conditional Gaussian limit]
\label{lem:marked-conditional}
Let
\[
(S_n,M_n)=\sum_\ell(X_{n\ell},Y_{n\ell}),
\qquad S_n\in\mathbb Z^d,\quad M_n\in\mathbb R^k,
\]
where $d\in\{1,2\}$ and $k$ are fixed, the joint increments are independent
and uniformly bounded, and a bounded tilt centers the constraint exactly at
$z_n$. Suppose, uniformly over the declared compact parameter class, that the
Fourier product after adjoining a complex mark frequency $iu/\sqrt n$ has the
same central Gaussian majorant and exponentially small minor arc for $u$ in
compact sets, and that
\[
\frac1{\sqrt n}\E_\lambda M_n\to m,
\qquad
\frac1n\Var_\lambda
\begin{pmatrix}S_n\\M_n\end{pmatrix}
\to
\begin{pmatrix}
\Sigma_{SS}&\Sigma_{SM}\\
\Sigma_{MS}&\Sigma_{MM}
\end{pmatrix},
\]
with $\Sigma_{SS}$ uniformly positive definite. Then, uniformly over the
declared compact parameter class,
\[
\mathcal L_\lambda\left(
\frac{M_n}{\sqrt n}\ \middle|\ S_n=z_n
\right)
\Longrightarrow
N\left(m,
\Sigma_{MM}-\Sigma_{MS}\Sigma_{SS}^{-1}\Sigma_{SM}
\right).
\]
For moment convergence, assume additionally that each bounded real mark tilt
$u/\sqrt n$ can be accompanied by an $O(n^{-1/2})$ change in the constraint
tilt that recenters $z_n$, while the joint Hessian, anchor probabilities,
covariance eigenvalue bounds, and displayed limits remain uniform. Under this
additional condition, conditional exponential moments are locally bounded and
the same conclusion holds for every fixed polynomial moment.
\end{lemma}

\begin{proof}
Fourier inversion is required only in the $d$ lattice coordinates. For fixed
$u$, the numerator of the conditional characteristic function is
\[
\frac1{(2\pi)^d}
\int_{[-\pi,\pi]^d}
\E_\lambda\exp\left[
it^\top(S_n-z_n)
+iu^\top\{M_n-\E_\lambda M_n\}/\sqrt n
\right]dt.
\]
Putting $t=s/\sqrt n$ gives the joint Gaussian expansion on the major arc;
the stipulated complex-frequency bound dominates the integral and removes the
minor arc. Dividing
by the local mass from Lemma~\ref{lem:uniform-lattice} gives the Schur-
complement covariance. Under the additional moment condition, a real
$O(n^{-1/2})$ mark tilt is accompanied by the stipulated $O(n^{-1/2})$
correction to the constraint tilt. The stable Hessian and local-mass bounds
then give uniform integrability and moment convergence.
\end{proof}

For the actual ARD arrays, the one-pair degree difference has anchor $1$.
The joint degree/$G$-tie array has anchors $(1,1)$ and $(1,0)$, whose
determinant is $-1$; the $H$ tie is identical. The shared-row triple has
anchors $(1,0)$ and $(0,1)$. For the fixed-path response, only degree is
conditioned on: the $G/H$ differences are marks, with joint support directions
$(1,0,0)$, $(1,1,0)$, and $(1,0,1)$. Positive group shares and compact
logistic predictors keep every anchor probability uniformly positive after
the finitely many focal deletions and bounded tilts. Bounded joint increments
and uniformly interior Bernoulli probabilities retain the anchor minor-arc
bound after compact complex mark frequencies. After the midpoint activity
tilt, cellwise sign balance leaves only a bounded mismatch in the degree
constraint, and an $O(n^{-1})$ common correction centers it exactly. Under a
bounded real mark tilt, the degree-constraint mean then moves by $O(\sqrt n)$;
the degree Hessian is order $n$, so an $O(n^{-1/2})$ constraint-tilt adjustment
recenters the target. The joint Hessian, anchors, covariance bounds, and marked
limits remain uniform under that adjustment.

\begin{lemma}[Collision and tie mass]
\label{lem:collision-mass}
There are constants $C,c_0>0$ such that, for every focal pair,
\[
 \Pp_c(D_i-D_j=d)
 \le Cn^{-1/2}\exp\{-c_0n(\alpha_i-\alpha_j)^2\},
\]
and
\[
 \begin{split}
 &\Pp_c(D_i=D_j,X_i=X_j)
 +\Pp_c(D_i=D_j,Z_i=Z_j)\\
 &\qquad\le Cn^{-1}\exp\{-c_0n(\alpha_i-\alpha_j)^2\}.
 \end{split}
\]
The matching lower local limit on $O(n^{-1/2})$ activity gaps implies
\[
 \E_cM_n^{\BO}\asymp n,
 \qquad
 \E_c(M_n^{\BO}-S_{2,n}^{\BO})=O(\sqrt n).
\]
\end{lemma}

\begin{proof}
Delete the focal nodes and exponentially tilt the Bernoulli difference array
to its midpoint.  Compact predictors and strong convexity of
$\log(1+e^u)$ give the affinity penalty
$\exp\{-c_0n(\alpha_i-\alpha_j)^2\}$. The one-constraint array uses the
anchor $1$; the joint total/$G$-tie array uses $(1,1)$ and $(1,0)$, and
similarly for $H$. Lemma~\ref{lem:uniform-lattice} therefore gives the
respective $n^{-1/2}$ and $n^{-1}$ factors. Because midpoint quantiles of a density
bounded above and away from zero satisfy
\[
 \sup_i\sum_j e^{-cn(\alpha_i-\alpha_j)^2}=O(\sqrt n),
 \qquad
 \sum_{i,j}e^{-cn(\alpha_i-\alpha_j)^2}=O(n^{3/2}),
\]
summing the local bounds gives the displayed upper orders. The exact
$10{:}10{:}20$ cell composition supplies positive $B/O$ mass in every local
activity neighborhood. On bounded scaled activity gaps, the exact logistic
overlap factor also gives the reverse bounded change-of-measure factor and the
tilted target remains in a bounded standardized region. The matching local
lower bound therefore gives the order-$n$ collision mass. A zero mark requires one of the two additional group-count
ties, and the two-constraint bound yields the $O(\sqrt n)$ tie defect.
\end{proof}

\begin{lemma}[Finite-state response range]
\label{lem:state-range}
For four distinct focal respondents, condition on the four cross dyads between
the two pairs.  If $m_{ij}(b)$ is the conditional mean of one pair kernel under
cross-edge state $b$, then
\[
 |m_{ij}(b)-m_{ij}(b')|
 \le Cn^{-1}\exp\{-c_0n(\alpha_i-\alpha_j)^2\}
\]
for any two cross-edge states $b,b'$.
\end{lemma}

\begin{proof}
After deleting the focal nodes, decompose the row difference into the three
disjoint blocks.  At the null, the exact logistic midpoint identity gives
\[
 \Pp(S_G=x,S_H=z,S_O=o)
 =C_{ij}e^{\eta(x+z+o)}q_G(x)q_H(z)q_O(o),
\]
where $\eta=(\alpha_i-\alpha_j)/2$, the $q_A$ are symmetric ternary
convolutions, and
$C_{ij}\le C e^{-c_0n(\alpha_i-\alpha_j)^2}$.  Uniform Fourier bounds give
$\|q_A\|_\infty=O(n^{-1/2})$ and
$\|\Delta q_A\|_\infty=O(n^{-1})$.  Changing a cross-edge state changes only
bounded degree and sign offsets, so telescoping over those offsets proves the
claim at $c=0$.

For bounded local $c$, factor each alter increment with its own logistic
midpoint tilt.  Normalize the resulting $O(n^{-1/2})$ tilts separately over
the nonempty $(G,H,O)\times\{+,-\}$ blocks.  The cellwise $20/20$ loading
balance keeps the affinity cross term bounded, while every observed group
contains order-$n$ uniformly interior span-one increments.  The same Fourier
and telescoping bounds therefore hold for the normalized multitype
convolutions, with the same Gaussian activity-gap envelope. Returning to the
original law changes the exponential prefactor by at most
$C\{|\alpha_i-\alpha_j|+n^{-1/2}\}$. Multiplying by the normalized
$O(n^{-1/2})$ mass and using
\[
|h|n^{-1/2}e^{-cnh^2}
\le Cn^{-1}e^{-c'nh^2}
\]
shows that the prefactor change has the same $n^{-1}$ order as the lattice
first difference. This completes the bounded-state comparison under the
original law.
\end{proof}

\begin{lemma}[Shared-row triple collision]
\label{lem:triple-collision}
For distinct $i,j,k$, after conditioning on their three internal dyads,
\[
 \begin{split}
 &\Pp_c(D_i-D_j=d_1,D_i-D_k=d_2)\\
 &\quad\le \frac Cn
 \exp[-c_0n\{(\alpha_i-\alpha_j)^2
                    +(\alpha_i-\alpha_k)^2\}].
 \end{split}
\]
\end{lemma}

\begin{proof}
Tilt the three external rows to their mean focal activity.  Strong convexity
on the compact predictor set supplies the displayed Gaussian penalty, up to a
bounded local-Gram term. The support directions $(1,0)$ and $(0,1)$ verify the
two-dimensional unimodular condition, so Lemma~\ref{lem:uniform-lattice}
gives a uniform $C/n$ point-mass bound. Mixing over the eight internal edge
states changes only bounded target offsets.
\end{proof}

\begin{lemma}[Network-center variance]
\label{lem:network-variance}
At $c=0$ and at the declared path $c=6$,
\[
 \Var_c(K_n^{\BO})=O(n),
 \qquad
 \Var_c(M_n^{\BO})=O(n).
\]
\end{lemma}

\begin{proof}
Decompose pair covariances by endpoint overlap.  Same-pair terms sum to at
most $\E S_{2,n}^{\BO}=O(n)$.  For pairs sharing one respondent, the product
of absolute marks is bounded by the triple degree-collision indicator;
Lemma~\ref{lem:triple-collision} and the packing bound give
\[
n^{-1}\sum_i
\left\{\sum_j e^{-cn(\alpha_i-\alpha_j)^2}\right\}^2
=O(n).
\]
For four distinct respondents, condition on the four cross
dyads.  The two residual kernels are then independent, and the
range-covariance inequality with Lemma~\ref{lem:state-range} bounds one
covariance by
\[
 Cn^{-2}\exp[-cn\{(\alpha_i-\alpha_j)^2
                         +(\alpha_k-\alpha_l)^2\}].
\]
Since
\[
n^{-2}\left\{\sum_{i,j}
e^{-cn(\alpha_i-\alpha_j)^2}\right\}^2=O(n),
\]
all four-distinct terms also sum to $O(n)$. The match-indicator proof is
identical without sign marks.
\end{proof}

\begin{lemma}[Graph-independent respondent transfer]
\label{lem:respondent-transfer}
Under either declared respondent design, at $c=0$ and $c=6$,
\[
 \Var_c(K_{R,n}^{\BO})=O(n),
 \qquad
 S_{2,R,n}^{\BO}/n
 \text{ is bounded away from zero in probability}.
\]
Equal-probability sampling also preserves both ratios of expected pair sums.
\end{lemma}

\begin{proof}
For pair $e=\{i,j\}$ write $I_e=R_iR_j$.  Graph independence gives the exact
identity
\[
 \Cov(I_eH_e,I_fH_f)
 =\E(I_eI_f)\Cov(H_e,H_f)
  +\Cov(I_e,I_f)\E H_e\E H_f.
\]
Bounded pair-inclusion probabilities preserve the three covariance classes
in Lemma~\ref{lem:network-variance}.  Under Bernoulli sampling, the additional
disjoint term is zero and the shared term is bounded by
\[
 \sum_i\left\{\sum_j\Pp(D_i=D_j)\right\}^2=O(n).
\]
Under simple random sampling, the $s$-node inclusion probability is
$\lambda_{s,n}=(r_n)_s/(n)_s$ and
$|\lambda_{4,n}-\lambda_{2,n}^2|=O(n^{-1})$.  Hence the added disjoint term is
\[
 O(n^{-1})\left\{\sum_{i<j}|\E h_{ij}|\right\}^2=O(n).
\]
Applying the same calculation to collision indicators and combining
Lemma~\ref{lem:collision-mass} with the $O(\sqrt n)$ tie defect gives the
denominator statement.  Finally, every eligible pair receives the same
inclusion factor, so it cancels from ratios of expected sums.
\end{proof}

\begin{lemma}[Conditional thinning normal approximation]
\label{lem:mask-clt}
Conditionally on the respondent frame and network, on any event where
$S_{2,R,n}^{\BO}$ is order $n$,
\[
 \frac{\sum_{i<j}(J_{ij}-\pi_n)W_{ij}h_{ij}}
 {\{\pi_n(1-\pi_n)S_{2,R,n}^{\BO}\}^{1/2}}
 \Rightarrow N(0,1).
\]
More precisely, its conditional Berry--Esseen error is bounded by a constant
times
\[
 \frac{(1-\pi_n)^2+\pi_n^2}
 {\{\pi_n(1-\pi_n)S_{2,R,n}^{\BO}\}^{1/2}}.
\]
\end{lemma}

\begin{proof}
The conditional summands are independent, centered, and bounded.  For one
nonzero mark,
\[
 \E|J-\pi|^3
 =\pi(1-\pi)\{(1-\pi)^2+\pi^2\}.
\]
Substitution into the Berry--Esseen ratio proves the bound.  It vanishes when
$\pi_n\to0$, $n\pi_n\to\infty$, and $S_{2,R,n}^{\BO}$ is order $n$.
\end{proof}

\begin{proof}[Proof of Theorem~\ref{thm:thinning-null}]
Write $S_2=S_{2,R,n}^{\BO}$ and $K_R=K_{R,n}^{\BO}$.  On
$\mathcal A_n=\{S_2\ge c n\}$, decompose
\[
 Z_n=W_n+B_n+C_n,
\]
where
\[
 W_n=\frac{\sum(J_{ij}-\pi_n)W_{ij}h_{ij}}
 {\{\pi_n(1-\pi_n)S_2\}^{1/2}},
\]
\[
 B_n=\sqrt{\frac{\pi_n}{1-\pi_n}}
 \frac{K_R-\E_0K_R}{\sqrt{S_2}},
 \qquad
 C_n=\sqrt{\frac{\pi_n}{1-\pi_n}}
 \frac{\E_0K_R}{\sqrt{S_2}}.
\]
Theorem~\ref{thm:finite-sign} gives $C_n\le0$.
Lemma~\ref{lem:respondent-transfer} gives
\[
 \E\{|B_n|\ind(\mathcal A_n)\}
 \le C\sqrt{\frac{\pi_n}{1-\pi_n}}.
\]
Combining Gaussian cdf Lipschitz continuity with
Lemma~\ref{lem:mask-clt} therefore yields
\[
 \begin{split}
 \Pp_0\{Z_n>z_{1-\alpha}\}
 \le{}&\alpha+C\sqrt{\frac{\pi_n}{1-\pi_n}}\\
 &+C\{n\pi_n(1-\pi_n)\}^{-1/2}
 +\Pp_0(\mathcal A_n^c).
 \end{split}
\]
Let $C_{R,n}$ be the collision count on the same frame and
$T_{R,n}=C_{R,n}-S_2$.  Lemmas~\ref{lem:collision-mass}--\ref{lem:respondent-transfer}
give constants $c_0,C_1,C_2>0$ with
\[
 \E C_{R,n}\ge c_0n,
 \quad \Var(C_{R,n})\le C_1n,
 \quad \E T_{R,n}\le C_2\sqrt n.
\]
The event inclusion
\[
 \{S_2<c_0n/4\}
 \subseteq
 \{C_{R,n}<c_0n/2\}\cup\{T_{R,n}>c_0n/4\}
\]
and Chebyshev plus Markov inequalities give
$\Pp_0(\mathcal A_n^c)=O(n^{-1/2})$.  For
$\pi_n=n^{-\beta}$ this proves the displayed envelope in the theorem.  In
particular, $\beta=1/2$ gives the conservative limit and balances the two
mask-generated rates.
\end{proof}

\section{Full-Pair Null Inference}
\label{app:full-pair-proof}

The proof of Theorem~\ref{thm:full-pair-null} establishes the group-specific
fixed-cell and sampled-nuisance conditions, the variance scale, the observable
denominator, and the centered normal limit.

\subsection{Fixed-cell intensities and the null row kernel}

For $g\in\{G,H,O\}$, let
\[
N_{g,n}(J)=\#\{i\in g:\alpha_i\in J\}.
\]
Recall the intensity $\lambda_g$ from
Assumption~\ref{ass:full-pair-profile}.

\begin{lemma}[Fixed-cell kernel and row projection]
\label{lem:full-pair-kernel}
For each fixed regular full-pair profile, uniformly over node labels,
\[
\sup_a\left|
N_{g,n}((-\infty,a])
-n\int_{-\infty}^a\lambda_g(t)\,dt
\right|=O(1),
\qquad
N_{g,n}(J)\le C\{1+n|J|\}.
\]
Put
\[
q_g(a)=\int\Lambda'(\mu+a+t)\lambda_g(t)\,dt,\qquad
q_+(a)=q_G(a)+q_H(a)+q_O(a).
\]
The functions $q_g$ are piecewise Lipschitz and uniformly bounded above and
away from zero on their relevant compact supports. At every common-support
continuity point, the degree-conditioned covariance of the $G/H$ count
differences is
\[
\Gamma(a)
=2\left[
\operatorname{diag}\{q_G(a),q_H(a)\}
-\frac{(q_G(a),q_H(a))^\top(q_G(a),q_H(a))}{q_+(a)}
\right],
\]
and is nonsingular. For a natural pair with
$t=\sqrt n(\alpha_j-\alpha_i)$,
\[
\sqrt n\,\E_0h_{ij}\longrightarrow K_0(a,t)
\]
uniformly on bounded $t$, where
\[
K_0(a,t)
=\frac{2\arcsin\{\rho_{GH}(a)\}/\pi}
{\sqrt{4\pi q_+(a)}}e^{-q_+(a)t^2/4},
\]
\[
\rho_{GH}(a)
=-
\sqrt{\frac{q_G(a)q_H(a)}
{\{q_H(a)+q_O(a)\}\{q_G(a)+q_O(a)\}}}.
\]
Globally,
\[
|\E_0h_{ij}|
\le Cn^{-1/2}e^{-cn(\alpha_i-\alpha_j)^2},
\qquad
\sum_j(\E_0h_{ij})^2\le Cn^{-1/2}.
\]
If $b_i=\sum_jA_{ij}^{\BO}\E_0h_{ij}$, then there are bounded functions
$\beta_B,\beta_O$, Lipschitz between the finitely many support endpoints, such
that
\[
b_i=\beta_{r(i)}(\alpha_i)+r_{i,n},\qquad
\frac1n\sum_i r_{i,n}^2\longrightarrow0.
\]
\end{lemma}

\begin{proof}
At $a=m_g(u)$, midpoint counting differs from $m_nu$ by at most one, which
gives the counting discrepancy and the interval bound after using
$m_g'\ge c_m$. Midpoint quadrature gives, uniformly in $a$,
\[
\frac1n\sum_{\ell\in g}\Lambda'(\mu+a+\alpha_\ell)
=q_g(a)+O(n^{-1}).
\]
After the activity-equalizing tilt, the normalized covariance of the degree
constraint and the two marks converges to
\[
2\begin{pmatrix}
q_+&q_G&q_H\\q_G&q_G&0\\q_H&0&q_H
\end{pmatrix}.
\]
The $O$, $G$, and $H$ blocks supply the unimodular directions
$(1,0,0),(1,1,0),(1,0,1)$. Lemmas
\ref{lem:uniform-lattice}--\ref{lem:marked-conditional} therefore apply after
the bounded focal deletions, and conditioning on degree takes the displayed
Schur complement. The bivariate Gaussian sign identity
$\E\{\sgn N_G\sgn N_H\}=2\arcsin(\rho_{GH})/\pi$ gives $K_0$.
The tilted local-mass and first-difference bounds give the global Gaussian
envelopes.

For $r=B,O$, define
\[
\beta_B(a)=\lambda_O(a)\int_{\R}K_0(a,t)\,dt,\qquad
\beta_O(a)=\{\lambda_G(a)+\lambda_H(a)\}
\int_{\R}K_0(a,t)\,dt.
\]
Away from a root-$n$ neighborhood of the finitely many support endpoints,
summation by parts converts the row sum into this local kernel integral.
The exceptional neighborhoods contain $O(\sqrt n)$ nodes, and the global
envelope controls the tails. The resulting error is $o(n)$ in squared
empirical norm. Focal deletion changes only lower-order terms, so the leading
kernel is common to focal $G$ and $H$ labels and pooling them as $B$ is valid.
\end{proof}

\subsection{Sampled row smoothing and the four-distinct remainder}

Let
\[
b_i^S=\E_0(s_i^S\mid S)
\]
and let $\rho_n=\Pp(R_j=1\mid R_i=1)$ for distinct nodes under the declared
respondent design. Then
\[
b_i^S=R_i\{\rho_nb_i+d_i^S\},\qquad
d_i^S=\sum_j(R_j-\rho_n)A_{ij}^{\BO}\E_0h_{ij}.
\]

\begin{lemma}[Sampled nuisance controls]
\label{lem:full-pair-nuisance}
There are graph-independent respondent-frame events
$\mathcal G_n^{(1)}$, with probability tending to one, such that uniformly
over $S\in\mathcal G_n^{(1)}$,
\[
\frac1n\|\widehat b^S-b^S\|_2^2
\longrightarrow_{P_0(\cdot\mid S)}0.
\]
Moreover, if
\[
V_{4,n}^S
=2\sum_{\substack{e<f\\e\cap f=\varnothing}}
\Cov_0(h_e,h_f\mid S),
\]
where the sum is over $\mathcal E_n^S$, then
\[
V_{4,n}^S=O_{P_S}(\sqrt n)=o_{P_S}(n).
\]
\end{lemma}

\begin{proof}
The envelope in Lemma~\ref{lem:full-pair-kernel} and the finite-population
second-moment formula give, for Bernoulli and fixed-size-SRS respondents,
\[
\E_S\left\{\frac1n\sum_iR_i(d_i^S)^2\right\}=O(n^{-1/2}).
\]
Let $Q_n(S)$ be the largest rolewise discrepancy between the sampled activity
counting process and its respondent-share multiple. On events of respondent
probability tending to one,
\[
Q_n(S)=O\{\sqrt{n\log n}\},
\]
all sampled role sizes are order $n$, and the preceding average squared mask
error is $o(1)$. Compact logits make expected full degree strictly increasing
in activity with slope of order $n$. Simultaneous degree concentration and
the interval-counting bound therefore give observed-degree rank displacement
$O_{P_0(\cdot\mid S)}\{\sqrt{n\log n}\}$.

For completeness, the localized covariance bound follows from an absolute
same-pair, shared-row, and four-distinct decomposition. For a deterministic
set $C$ contained in one role, put
\[
R_C^S=\sum_{i\in C}(s_i^S-b_i^S).
\]
The one-collision envelope gives
\[
\E_0h_{ij}^2
\le Cn^{-1/2}\exp\{-cn(\alpha_i-\alpha_j)^2\}.
\tag{B.1}
\]
Because the Gaussian partner sum is $O(\sqrt n)$ uniformly in its anchor,
the same-pair contribution to $\Var_0(R_C^S\mid S)$ is $O(|C|)$. For a
shared row,
\[
\E_0|h_{ij}h_{ik}|
\le Cn^{-1}\exp\{-cn[(\alpha_i-\alpha_j)^2
                         +(\alpha_i-\alpha_k)^2]\}.
\tag{B.2}
\]
If the shared endpoint belongs to $C$, direct rowwise summation is
$O(|C|)$. If the two leaves belong to $C$, define
\[
K_C(a)=\sum_{i\in C}\exp\{-cn(\alpha_i-a)^2\}.
\]
The fixed-profile packing bound gives
\[
\sup_aK_C(a)\le C\min(|C|,\sqrt n),
\qquad
\sum_{k=1}^nK_C(\alpha_k)\le C|C|\sqrt n,
\]
and hence $n^{-1}\sum_kK_C(\alpha_k)^2\le C|C|$. The
product-of-means subtraction is no larger.

For four distinct respondents, write $e=\{u,v\}$, $f=\{w,x\}$ and condition
on the four cross dyads
\[
\mathcal Z_{ef}=\{A_{uw},A_{ux},A_{vw},A_{vx}\}.
\]
For $B=e,f$, let $F_B(\mathcal Z_{ef})$ be the conditional pair contribution.
Its product-Bernoulli ANOVA expansion is
\[
F_B-\E F_B
=\sum_{\varnothing\ne R\subseteq\mathcal Z_{ef}}
\widehat F_B(R)\prod_{a\in R}(A_a-p_a),
\]
and orthogonality gives
\[
\Cov(F_e,F_f)
=\sum_{\varnothing\ne R\subseteq\mathcal Z_{ef}}
\widehat F_e(R)\widehat F_f(R)
\prod_{a\in R}p_a(1-p_a).
\]
With $\Delta_R=\prod_{a\in R}\Delta_a$, the coefficient is exactly
\[
\widehat F_B(R)
=\E_{\mathcal Z_{ef}\setminus R}
\{\Delta_RF_B(\mathcal Z_{ef}\setminus R)\}.
\tag{B.3a}
\]
After the focal cross edges are deleted, $F_B$ is a bounded complete pair
factor with one degree-difference constraint. Restoring the $r=|R|$ selected
cross edges gives $r$ coherent bounded lattice shifts. Applying
Lemma~\ref{lem:full-pair-complete-response} below with $d=1$ and $t=r$, and
using the ordered product rule in (B.23)--(B.24), gives, for
$B=\{y_B,z_B\}$,
\[
|\widehat F_B(R)|
\le C_rn^{-(r+1)/2}
\exp\{-cn(\alpha_{y_B}-\alpha_{z_B})^2\}.
\tag{B.3}
\]
Writing
\[
G_C=\sum_{\{y,z\}:\{y,z\}\cap C\ne\varnothing}
\exp\{-cn(\alpha_y-\alpha_z)^2\},
\]
we have $G_C\le C|C|\sqrt n$. The support-$r$ aggregate is at most
\[
C_rn^{-(r+1)}G_C^2\le C_r|C|^2n^{-r}.
\tag{B.4}
\]
For $r=1$, $|C|^2/n\le|C|$; every higher support is smaller. Conditional on
a fixed respondent frame, the mask only replaces pair weights by zero. Thus
\[
\Var_0(R_C^S\mid S)\le C|C|.
\tag{B.5}
\]
A broad block is split into its disjoint $G$ and $H$ pieces and bounded by
\[
\Var(X_G+X_H)\le2\Var(X_G)+2\Var(X_H),
\]
so no cross-role cancellation is used. The leave-one-out identity,
Lemma~\ref{lem:full-pair-kernel}, and the pathwise partition inequality yield
\[
\frac1n\|\widehat b^S-b^S\|_2^2
=O_{P_0(\cdot\mid S)}\left[
\ell_n^{-1}
+\frac{\sqrt{n\log n}}{\ell_n}
+\left\{\frac{\ell_n+Q_n}{n}\right\}^2
+\frac{\ell_n}{n}+\varepsilon_n
\right],
\]
where $\varepsilon_n\to0$. The block conditions make every term vanish.

For the global four-distinct remainder, the support orders after pair deletion
are
\[
\begin{array}{c|c|c}
r&\text{coefficient envelope}&\text{aggregate contribution}\\ \hline
1&\text{exact identity below}&O_{P_S}(\sqrt n)\\
2&O(n^{-3/2})&O(1)\\
3&O(n^{-2})&O(n^{-1})\\
4&O(n^{-5/2})&O(n^{-2}).
\end{array}
\tag{B.6a}
\]
The support-one terms are the four singleton subsets of
$\mathcal Z_{ef}$. For a graph edge $\{u,v\}$, let $r_{uw}^{uv}$ be the
change in the $\{u,w\}$ pair contribution when $A_{uv}$ is toggled from zero
to one, extended by zero outside its natural partner domain, and put
$R_{u\mid v}^S=R_u\sum_wR_wr_{uw}^{uv}$. Summing singleton terms gives
\[
\begin{split}
V_{4,n}^{S,(1)}
=2\sum_{u<v}p_{uv}(1-p_{uv})\bigg[
&R_{u\mid v}^SR_{v\mid u}^S\\
&-R_uR_v\sum_wR_wr_{uw}^{uv}r_{vw}^{uv}
\bigg].
\end{split}
\tag{B.6}
\]
Put
\[
R_{u\mid v}=\sum_wr_{uw}^{uv},
\qquad
\xi_{u\mid v}^S=\sum_w(R_w-\rho_n)r_{uw}^{uv},
\]
so that $R_{u\mid v}^S=R_u\{\rho_nR_{u\mid v}+\xi_{u\mid v}^S\}$.
For Bernoulli sampling, use the independent-mask variance formula. For
fixed-size SRS, condition on $R_u=1$ and use the centered finite-population
variance formula for the remaining sample. In both cases,
\[
\E_S\{R_u(\xi_{u\mid v}^S)^2\}
\le C\sum_w(r_{uw}^{uv})^2
\le Cn^{-3/2}.
\tag{B.6b}
\]
The deterministic reflection calculation gives
\[
|R_{u\mid v}|
\le Cn^{-1}
+Cn^{-1/2}\sum_{a\in\mathcal B}
e^{-cn d(\alpha_u,a)^2},
\tag{B.6c}
\]
where $\mathcal B$ is the finite endpoint set. The profile packing bound,
(B.6b), and Markov's inequality imply
\[
\sum_{u<v}\{R_{u\mid v}^2+R_{v\mid u}^2\}=O(\sqrt n),
\qquad
\sum_{u<v}\{(\xi_{u\mid v}^S)^2
 +(\xi_{v\mid u}^S)^2\}=O_{P_S}(\sqrt n).
\tag{B.6d}
\]
Expanding the first product in (B.6) and applying Cauchy--Schwarz to the three
combinations of deterministic and random terms gives $O_{P_S}(\sqrt n)$.
The subtraction is $O(1)$ absolutely by the
three-gap majorant. For supports two through four there are respectively
$6$, $4$, and $1$ subsets of the four-edge cube. The finite support
multiplicities and the coefficient envelopes in (B.6a) make their label sums
absolutely summable at the displayed orders. This proves the second assertion
without assuming reflected-cell cancellation after respondent sampling.
\end{proof}

\subsection{Variance scale and observable denominator}

\begin{lemma}[Directed-row variance restoration]
\label{lem:full-pair-variance-restoration}
For each fixed regular full-pair profile, there are graph-independent
respondent-frame events of probability tending to one on which
\[
\Var_0(U_n^S\mid S)\ge cn.
\]
The constant may depend on the fixed profile through an interior
broad/outside overlap margin.
\end{lemma}

\begin{proof}
Fix a respondent frame $S$. For ordered $i\ne\ell$, let
$Z_{i\ell}\sim\operatorname{Bernoulli}(p_{i\ell})$ independently over ordered
coordinates. For $e=\{i,j\}$, define $\widetilde h_e$ from the two independent
rows $i,j$ after setting $Z_{ij}=Z_{ji}=0$, and let
$\widetilde U_n^S=\sum_{e\in\mathcal E_n^S}\widetilde h_e$. Its
nonidentical-row Hoeffding decomposition is exact:
\[
\Var(\widetilde U_n^S\mid S)
=\sum_{e\in\mathcal E_n^S}\E r_e^2
+\sum_i\E\left(\sum_{e\ni i}\phi_e^{(i)}\right)^2.
\tag{B.7}
\]
Both terms are nonnegative.

We first compare the actual and directed-row same/shared contributions. For a
pair square, delete the mutual focal dyad; for a wedge
$e=\{i,j\},f=\{i,k\}$, delete every focal-to-focal coordinate among
$\{i,j,k\}$. After this complete focal deletion, the actual external row
stars and the comparator external rows have the same product law.

For $q=2$ in a pair square and $q=3$ in a wedge, restoring one actual or
directed focal coordinate changes a complete expectation by at most
\[
Cn^{-q/2}\exp\{-cnQ_q(\alpha)\},
\tag{B.8}
\]
where $Q_2\asymp(\alpha_i-\alpha_j)^2$ and
\[
Q_3\asymp(\alpha_i-\alpha_j)^2+(\alpha_i-\alpha_k)^2.
\]
Indeed, after focal deletion there are $q-1$ root-row degree contrasts. Every
queried group supplies a fixed unimodular anchor, so the covariance is order
$n$ in dimensions one and two. A bounded degree-target shift costs one
lattice difference, while a sign-threshold shift fixes one primitive group
contrast and convolves the remaining blocks. Both operations have order
$n^{-q/2}$. The activity-equalizing tilt supplies the Gaussian gap envelope;
telescoping over finitely many restored coordinates preserves the order.

Consequently,
\[
|\E_0h_{ij}^2-\E\widetilde h_{ij}^2|
\le Cn^{-1}e^{-cn(\alpha_i-\alpha_j)^2},
\tag{B.9}
\]
and
\[
|\E_0(h_{ij}h_{ik})-
\E(\widetilde h_{ij}\widetilde h_{ik})|
\le Cn^{-3/2}e^{-cnQ_3(\alpha)}.
\tag{B.10}
\]
The comparison remains valid after centering because
\[
|\E_0h_{ij}|+|\E\widetilde h_{ij}|
\le Cn^{-1/2}e^{-cn(\alpha_i-\alpha_j)^2},
\]
\[
|\E_0h_{ij}-\E\widetilde h_{ij}|
\le Cn^{-1}e^{-cn(\alpha_i-\alpha_j)^2},
\]
and
\[
|\mu_e\mu_f-\widetilde\mu_e\widetilde\mu_f|
\le|\mu_e-\widetilde\mu_e|\,|\mu_f|
+|\widetilde\mu_e|\,|\mu_f-\widetilde\mu_f|.
\tag{B.11}
\]
The fixed-profile packing sums
\[
\sum_{i,j}e^{-cn(\alpha_i-\alpha_j)^2}=O(n^{3/2}),
\qquad
\sum_{i,j,k}e^{-cnQ_3(\alpha)}=O(n^2)
\]
therefore give
\[
\left|V_{SS,n}^{S,\mathrm{actual}}
-\Var(\widetilde U_n^S\mid S)\right|=O(\sqrt n),
\tag{B.12}
\]
where $V_{SS,n}^{S,\mathrm{actual}}$ contains the actual same-pair and
shared-row covariance classes. Disjoint comparator kernels use disjoint
independent rows and have zero covariance.

It remains to lower-bound the first term in (B.7). Positive $\Omega_{BO}$
yields one role $g_*\in\{G,H\}$ and a compact interval $J$ away from the
profile breakpoints on which both $\lambda_{g_*}$ and $\lambda_O$ are bounded
below. For sufficiently small fixed $\delta>0$, every natural $g_*/O$ pair in
$J$ satisfying $|\alpha_i-\alpha_j|\le\delta n^{-1/2}$ obeys
\[
\Pp(D_i^-=D_j^-)\ge cn^{-1/2},
\tag{B.13}
\]
where $D_i^-$ is the comparator degree after bounded focal deletion. The two
events that additionally tie either queried group count have total probability
$O(n^{-1})$, so $\E\widetilde h_{ij}^2\ge c_1n^{-1/2}$. Conditional on either
complete comparator row, the partner degree is a sum of $n-O(1)$ independent
Bernoulli variables with probabilities in a common compact subinterval of
$(0,1)$, and hence
\[
\sup_d\Pp(D_j^-=d)\le Cn^{-1/2}.
\]
With $\theta_{ij}=\E\widetilde h_{ij}$ and
$\phi_{ij}^{(i)}=\E(\widetilde h_{ij}\mid Z_i)-\theta_{ij}$, write
\[
r_{ij}=\widetilde h_{ij}-\theta_{ij}
-\phi_{ij}^{(i)}-\phi_{ij}^{(j)}.
\]
The two row projections and the mean have squared order $O(n^{-1})$.
Orthogonality gives
\[
\E r_{ij}^2
=\E\widetilde h_{ij}^2-\theta_{ij}^2
-\E\{\phi_{ij}^{(i)}\}^2
-\E\{\phi_{ij}^{(j)}\}^2
\ge c_2n^{-1/2}
\tag{B.14}
\]
uniformly over these close pairs.

Midpoint quadrature gives $c_Jn^{3/2}\{1+o(1)\}$ such deterministic close
pairs. Bernoulli respondent sampling retains order $n^{3/2}$ of them because
the band graph has maximum degree $O(\sqrt n)$ and retained-count variance
$O(n^2)$. Fixed-size SRS has the same order: overlapping-edge terms are
$O(n^2)$, and the $O(n^{-1})$ covariance over $O(n^3)$ disjoint edge pairs is
also $O(n^2)$. On graph-independent events of probability tending to one,
(B.7) and (B.14) imply
\[
\Var(\widetilde U_n^S\mid S)\ge cn.
\tag{B.15}
\]
Finally,
\[
\Var_0(U_n^S\mid S)
=V_{SS,n}^{S,\mathrm{actual}}+V_{4,n}^S.
\]
Equations (B.12), (B.15), and
$V_{4,n}^S=O_{P_S}(\sqrt n)=o_{P_S}(n)$ prove the result.
\end{proof}

\begin{lemma}[Full-pair information and denominator]
\label{lem:full-pair-denominator}
Under Assumption~\ref{ass:full-pair-profile}, there are
graph-independent events $\mathcal G_n^{(2)}\subseteq\mathcal G_n^{(1)}$,
with probability tending to one, on which
\[
cn\le\sigma_{n,S}^2\le Cn
\]
and
\[
\widehat V_n^S-\sigma_{n,S}^2
=o_{P_0(\cdot\mid S)}(n)
\]
uniformly in $S$. Consequently,
$\widehat V_n^S/\sigma_{n,S}^2\to1$ conditionally and uniformly on these
frames.
\end{lemma}

\begin{proof}
Lemma~\ref{lem:full-pair-variance-restoration} gives the lower variance bound
on graph-independent good-frame events. The collision and packing envelopes
give the matching $O(n)$ upper bound.

Connected label sums obey
\[
\sum_{\text{labels}}e^{-cnQ(\alpha)}
=O\{n^{(v+c)/2}\}
\]
for every fixed shape with $v$ labels and $c$ components: choose one free
anchor per component and use the $O(\sqrt n)$ local packing sum for each
remaining label. The complete pair-square and wedge expansions therefore give
\[
\Var_0(A_n^S\mid S)=O(n),\qquad
\Var_0(T_n^S\mid S)=O(n).
\]
For $\mu_e=\E_0(h_e\mid S)$, direct expansion gives the exact bias identity
\[
\sigma_{n,S}^2
-\E_0\left[
A_n^S+T_n^S-\sum_i(b_i^S)^2\,\middle|\,S
\right]
=V_{4,n}^S+\sum_e\mu_e^2.
\]
The Gaussian pair envelope gives $\sum_e\mu_e^2=O(\sqrt n)$.
Lemma~\ref{lem:full-pair-nuisance}, concentration of $A_n^S,T_n^S$, and
Cauchy--Schwarz for the row-square replacement prove the additive $o_p(n)$
denominator error. The variance lower bound converts it to ratio
consistency.
\end{proof}

\begin{lemma}[Fixed-dimensional complete response]
\label{lem:full-pair-complete-response}
Fix a dimension $d$, response order $t$, and bounded lattice directions
$v_1,\ldots,v_t$. Let $X_{n,\ell}$ be independent $\mathbb Z^d$-valued
increments, with their number between positive constant multiples of $n$.
If, uniformly in $n,\ell$,
\[
\Pp(X_{n,\ell}=0)\ge\eta,
\qquad
\Pp(X_{n,\ell}=e_j)\ge\eta,
\quad j=1,\ldots,d,
\tag{B.16}
\]
then the pmf $q_n$ of their sum satisfies
\[
\|\Delta_{v_1}\cdots\Delta_{v_t}q_n\|_1
\le C_{d,t}n^{-t/2},
\qquad
\|\Delta_{v_1}\cdots\Delta_{v_t}q_n\|_\infty
\le C_{d,t}n^{-(d+t)/2}.
\tag{B.17}
\]
\end{lemma}

\begin{proof}
Let $\nu_d$ be uniform on $\{0,e_1,\ldots,e_d\}$. Condition (B.16) gives a
fixed minorization
\[
\mathcal L(X_{n,\ell})=\theta\nu_d+(1-\theta)\rho_{n,\ell}.
\tag{B.18}
\]
Independent selectors produce, with exponentially high probability, an anchor
convolution factor $\nu_d^{*N}$ with $N\asymp n$. Convolution with the residual
sum contracts both norms after all finite differences are placed on this
factor.

The anchor characteristic function is
\[
\phi_d(u)=\frac{1+\sum_{j=1}^de^{iu_j}}{d+1},
\]
and strong aperiodicity gives
\[
|\phi_d(u)|\le
\exp\{-c_d\operatorname{dist}(u,2\pi\mathbb Z^d)^2\}.
\tag{B.19}
\]
For $a(u)=\prod_{j=1}^t(e^{iv_j^Tu}-1)$,
$|a(u)|\le C\|u\|^t$. Fourier inversion and (B.19) give the
$\ell^\infty$ half of (B.17). For $\ell^1$, choose an integer $L>d/2$ and
apply Cauchy--Schwarz with weight
$(1+\|x-N\E_{\nu_d}V\|^2/N)^L$. Parseval converts the weighted squared norm
to derivatives of
\[
e^{-iN(\E_{\nu_d}V)^Tu}a(u)\phi_d(u)^N.
\]
After scaling derivative order $|\gamma|$ by $N^{-|\gamma|/2}$, each term is
bounded by
\[
CN^{-t/2}
\{1+\sqrt N\operatorname{dist}(u,2\pi\mathbb Z^d)\}^{C_{L,t}}
e^{-cN\operatorname{dist}(u,2\pi\mathbb Z^d)^2}.
\]
Integration makes the weighted $\ell^2$ factor
$O(N^{-d/4-t/2})$, while the inverse weight has $\ell^2$ norm
$O(N^{d/4})$. This proves the $\ell^1$ half. Averaging over the selector count
and absorbing its exponential lower tail proves (B.17).
\end{proof}

\begin{lemma}[Fixed-order full-pair cumulants]
\label{lem:full-pair-fixed-cumulants}
For each fixed regular full-pair profile and every fixed $r\ge3$, uniformly on
the good graph-independent respondent frames,
\[
 |\kappa_r(U_n^S\mid S)|\le C_rn.
 \tag{B.20}
\]
\end{lemma}

\begin{proof}
Consider one complete local moment factor after deleting its finite focal
alter set. Choose one root in every respondent component and let $d$ be the
total root-contrast rank. Under one corner-independent product tilt, the
independent nonfocal group sums are $R_G,R_H,R_O\in\mathbb Z^d$. In the
unimodular coordinates
\[
 (S,M)=(R_G+R_H+R_O,R_G,R_H),
\]
their unconditioned joint pmf, evaluated on the constraint slice $S=s$, is
\[
 p_n(s,g,h)=q_G(g)q_H(h)q_O(s-g-h).
 \tag{B.21}
\]

Every nonfocal group has order-$n$ size. Compact tilted edge probabilities
and the configurations with all focal-row edges zero or exactly one nonroot
edge present give (B.16) in each group. A joint response of order $t$ splits,
by the ordered discrete product rule, into group response orders
$r_G+r_H+r_O=t$. Put one factor in $\ell^\infty$ and the other two in
$\ell^1$. Lemma~\ref{lem:full-pair-complete-response} yields
\[
 \sup_s\sum_{g,h}
 |\Delta^{r_G}q_G(g)\Delta^{r_H}q_H(h)
   \Delta^{r_O}q_O(s-g-h)|
 \le Cn^{-(d+t)/2}.
 \tag{B.22}
\]
This product is formed before conditioning; no conditional independence of
the three group sums is used.

For a fixed raw/ghost decoration, let $(S_n,M_n)$ be the outside lattice sum
after deleting the finite focal dyads. The coordinates of $M_n$ are the
primitive count contrasts: repeated sign factors using the same contrast
share one coordinate. If $z\in\{0,1\}^q$ is a bridge corner, write $Hz$ for
its degree-constraint translation and $Cz$ for its coherent primitive-mark
translation. For a bounded complete mark $w$,
\[
 F(z)=\E\left[
 w(M_n+Cz)\mathbf1\{S_n=b+Hz\}
 \right].
 \tag{B.23}
\]
If $p_n^{\mathrm{orig}}(s,m)=\Pp(S_n=s,M_n=m)$, the change of variables
$\ell=m+Cz$ gives the exact complete-translation identity
\[
 F(z)=\sum_{\ell}w(\ell)
 p_n^{\mathrm{orig}}(b+Hz,\ell-Cz),
 \qquad |w|\le1.
 \tag{B.23a}
\]
Writing $h_j,c_j$ for the columns of $H,C$ and
$\Delta_vf(x)=f(x+v)-f(x)$, every mixed Boolean response is therefore
\[
 \Delta_TF(0)
 =\sum_\ell w(\ell)
 \left\{\prod_{j\in T}\Delta_{(h_j,-c_j)}
 p_n^{\mathrm{orig}}\right\}(b,\ell).
 \tag{B.23b}
\]
All appearances of one primitive contrast share a single coordinate and
receive the same edge-induced translation. To return from the tilted law,
write the component tilt in node-potential coordinates. For a
respondent/ghost component $C$, the potential
$t_x=\bar\alpha_C-\alpha_x$ centers its outside root contrasts. A fixed
forest gauge gives
\[
 \|\lambda\|\le C\{Q_F(\alpha)^{1/2}+n^{-1}\},
\]
while strong convexity gives, for every bounded corner shift $u$,
\[
 L_\lambda(b+u)\le Ce^{-cnQ_F(\alpha)}.
\]
If $k$ responses hit this likelihood factor, then
\[
 \{Q_F^{1/2}+n^{-1}\}^ke^{-cnQ_F}
 \le C_kn^{-k/2}e^{-c'nQ_F}.
\]
Combining these charges with (B.22) gives the complete-factor response
\[
 |\Delta_TF|
 \le C_r n^{-\{d(F)+|T|\}/2}e^{-cQ_F(\alpha)n}.
 \tag{B.24}
\]
The sign product enters (B.23) as the bounded weight $w$.

The raw/ghost reduction entering this complete-factor argument is finite and
exact. After the bounded focal deletion, let $P_i$ denote
conditional expectation over endpoint row $i$. For one pair occurrence
$e=\{i,j\}$, its pair-Hoeffding remainder is
\[
 r_e=(I-P_i)(I-P_j)h_e.
\]
If $X_i',X_j'$ are independent ghost rows with the same activities as
$X_i,X_j$, then
\[
 r_e(X_i,X_j)
 =\E_{X_i',X_j'}\{h_e(X_i,X_j)-h_e(X_i',X_j)
                  -h_e(X_i,X_j')+h_e(X_i',X_j')\}.
\]
Consequently, a product of at most $r$ such occurrences is a signed linear
combination of at most $4^r$ raw/ghost decorations; a raw $h_e$ is simply the
all-raw decoration. In every decoration, each original occurrence edge is
preserved bijectively in a split graph $F$, and a map
$\pi:V(F)\to V(H)$ sends each ghost endpoint to its source respondent, so the
ghost activity is $\alpha_{\pi(x)}$. Retained copies of one physical bridge
continue to use the same Bernoulli coordinate rather than independent copies.
If every occurrence of a bridge is integrated out by ghost replacement, its
Boolean response is zero. The number of decorations, deletion patterns, and
corner states is therefore bounded by a constant depending only on fixed
$r$.

For pair-kernel occurrences $h_{e_1},\ldots,h_{e_r}$, form the respondent
multigraph $H$ whose labeled vertices are the distinct respondents and whose
edge multiset is $(e_1,\ldots,e_r)$. The joint cumulant is
\[
 \kappa(h_{e_1},\ldots,h_{e_r})
 =\sum_{\pi\in\Pi_r}(|\pi|-1)!(-1)^{|\pi|-1}
   \prod_{B\in\pi}\E\prod_{j\in B}h_{e_j}.
 \tag{B.24a}
\]
The cancellation that will be used below is an exact finite identity. Let
$I_a$ be the labeled occurrence slots in respondent component $a$, let $L_a$
collect its local randomness, and let the cross-component Bernoulli bridge
dyads be $Z_b$, mutually independent and independent of the $L_a$. For each
nonempty $C\subseteq I_a$, define the conditional local cumulant
\[
 K_C(Z)=\kappa_{L_a}(h_{e_s}:s\in C\mid Z).
\]
Pointwise moment--cumulant inversion and the product-Bernoulli ANOVA expansion
give
\[
 \E_{L_a}\left(\prod_{s\in A}h_{e_s}\middle|Z\right)
 =\sum_{\lambda\in\Pi(A)}\prod_{C\in\lambda}K_C(Z),
 \qquad
 K_C(Z)=\sum_{T\subseteq D_a}\widehat K_C(T)
             \prod_{b\in T}(Z_b-\E Z_b),
 \tag{B.24b}
\]
where $D_a$ is the set of bridges incident to component $a$ and
$\widehat K_C(T)$ is the product-measure average of $\Delta_TK_C$.

A complete diagram $d$ consists of a local-atom partition $\lambda$, one
support $T_C$ for each atom, and, for every physical bridge $b$, a partition
$\gamma_b$ of the atoms selecting $b$. A block $G\in\gamma_b$ contributes
the primitive Bernoulli cumulant $\kappa_{|G|}(Z_b-\E Z_b)$. Let $\rho(d)$
be the partition of the labeled slots generated by joining slots in the same
local atom and slots whose atoms lie in the same primitive bridge block, and
put
\[
 w(d)=\prod_{C\in\lambda}\widehat K_C(T_C)
       \prod_b\prod_{G\in\gamma_b}
       \kappa_{|G|}(Z_b-\E Z_b).
\]
For every raw moment partition $\pi$, the two inversions in (B.24b) give
\[
 \prod_{B\in\pi}\E\prod_{s\in B}h_{e_s}
 =\sum_{d:\rho(d)\le\pi}w(d).
\]
Substituting this into (B.24a) and using the partition-lattice identity
\[
 \sum_{\pi:\rho\le\pi\le\widehat1}
 (-1)^{|\pi|-1}(|\pi|-1)!
 =\mathbf1\{\rho=\widehat1\}
\]
proves the exact connected-diagram formula
\[
 \kappa(h_{e_1},\ldots,h_{e_r})
 =\sum_{d:\rho(d)=\widehat1}w(d).
 \tag{B.24c}
\]
For fixed $r$ this is a finite sum. Expand every component-local atom in
(B.24c) into complete moment factors $F_{C,j}$ before applying the ordered
product rule. Their response subsets partition the incident responses, and
their occurrence-edge multisets partition the original respondent component,
with ghost vertices carrying their source activities. For a finite
multigraph $F$, let $B_F$ be its incidence matrix and put
\[
 d(F)=\operatorname{rank}B_F,
 \quad
 Q_F=\sum_{C\in\operatorname{cc}(F)}
       \sum_{x\in C}(\alpha_x-\bar\alpha_C)^2,
 \quad
 D_F=\sum_{\{x,y\}\in E(F)}(\alpha_x-\alpha_y)^2,
\]
counting repeated edge occurrences. If $H_a$ has $v_a$ labels, rank
subadditivity gives
\[
 v_a-1=\operatorname{rank}B_{H_a}
 \le\sum_{C,j}\operatorname{rank}B_{F_{C,j}}.
 \tag{B.25a}
\]
Every factor has at most $r$ edges, its edge multiset partitions that of
$H_a$, and a spanning-tree path in $H_a$ has length at most $r$. Therefore
\[
 Q_{H_a}
 \le\frac{r(r+1)}2D_{H_a},
 \qquad
 D_{H_a}=\sum_{C,j}D_{F_{C,j}},
 \qquad
 D_{F_{C,j}}\le2rQ_{F_{C,j}}.
 \tag{B.25b}
\]
Equations (B.25a)--(B.25b) give explicitly
\[
 \sum_{C,j}d(F_{C,j})\ge v_a-1,
 \qquad
 \sum_{C,j}Q_{F_{C,j}}
 \ge\{r^2(r+1)\}^{-1}Q_{H_a}.
 \tag{B.25}
\]
Thus rank and energy are aggregated at component level; no individual atom is
asserted to have rank $v_a-1$.

Moreover, after weakening the fixed response constant if necessary, (B.25)
implies
\[
 \prod_{a,C,j}\exp\{-cnQ_{F_{C,j}}\}
 \le \exp\left\{-c_rn\sum_{a=1}^cQ_{H_a}\right\}
 =:\mathcal E_H(\boldsymbol i),
 \tag{B.25d}
\]
where $\boldsymbol i$ is the legal respondent-label assignment for the fixed
shape $H$.

For a global diagram with $v$ labels and $c$ respondent components, the
condition $\rho(d)=\widehat1$ in (B.24c) makes its contracted bridge graph
connected. It therefore contains at least $c-1$ distinct physical bridge
dyads. A singleton primitive bridge block has zero weight because
$\E(Z_b-\E Z_b)=0$; each selected tree bridge must consequently be selected
by at least one local atom at both physical endpoints. If $q_a$ counts the
distinct connecting bridge coordinates selected in component $a$, then
\[
 \sum_{a=1}^cq_a\ge2(c-1).
 \tag{B.25c}
\]
Thus connected M\"obius--ANOVA cancellation supplies at least $2(c-1)$
response charges. Equations (B.24)--(B.25c) bound one fixed shape by
\[
 C_rn^{-(v-c)/2-(c-1)}\mathcal E_H.
 \tag{B.26}
\]
Summing the envelope in (B.25d) over legal label assignments gives
$O_r\{n^{(v+c)/2}\}$: choose one free anchor per component and use the
$O(\sqrt n)$ local packing sum for every remaining label. Hence the total
exponent is
\[
 -\frac{v-c}{2}-(c-1)+\frac{v+c}{2}=1,
\]
which proves (B.20).
\end{proof}

\begin{lemma}[Centered full-pair normal limit]
\label{lem:full-pair-clt}
Uniformly on the good frames in
Lemma~\ref{lem:full-pair-denominator},
\[
\frac{U_n^S-\E_0(U_n^S\mid S)}{\sigma_{n,S}}
\Rightarrow N(0,1).
\]
\end{lemma}

\begin{proof}
Lemma~\ref{lem:full-pair-fixed-cumulants} and
Lemma~\ref{lem:full-pair-denominator} give, for every fixed $r\ge3$,
\[
 \left|
 \kappa_r\left(
 \frac{U_n^S-\E_0(U_n^S\mid S)}{\sigma_{n,S}}
 \middle|S\right)
 \right|
 \le C_rn^{1-r/2}\longrightarrow0.
\]
The normalized first two cumulants are exactly $0$ and $1$. The
moment--cumulant formula therefore converts the preceding fixed-order
cumulant limits into convergence of every fixed moment to the corresponding
standard Gaussian moment. The unit second moments make the sequence tight,
and the Gaussian law is moment-determinate; the method of moments gives the
centered conditional normal limit on every good-frame sequence. A
violating-frame subsequence argument upgrades this to the stated uniform
conditional convergence. Each moment calculation uses a fixed cumulant order
and a fixed-dimensional local limit.
\end{proof}

\begin{proof}[Proof of Theorem~\ref{thm:full-pair-null}]
Let $\mathcal G_n$ be the intersection of the good-frame events in
Lemmas~\ref{lem:full-pair-nuisance}--\ref{lem:full-pair-denominator}.
Its respondent-design probability tends to one. The exact finite sign
restriction gives
\[
\E_0(U_n^S\mid S)\le0
\]
for every such frame at each finite $n$.
Lemmas~\ref{lem:full-pair-denominator}--\ref{lem:full-pair-clt} give the
centered studentized normal limit. Since
$a_n\to0$, $b_n\to\infty$, and $\sigma_{n,S}^2\asymp n$, clipping is inactive
under the null with conditional probability tending to one uniformly over the
good frames. Replacing the unknown nonpositive center by zero can only reduce
the positive-tail rejection probability. This proves the conditional bound;
averaging over the respondent frame proves the joint-law statement.
\end{proof}

\section{Finite Bivariate Nonredundancy}
\label{app:finite-witness}

\begin{proof}[Proof of Proposition~\ref{prop:nonredundancy}]
Take focal roles $i\in G$ and $j\in O$, with
$\Pp(A_{ij}=1)=1/5$. Use five fixed core-alter slots and two fixed outside
slots. The outside probabilities are unchanged across constructions:
\[
 p_{iO}=(7/20,9/20),\qquad p_{jO}=(2/5,3/10).
\]
The row-by-slot core probabilities are
\[
\begin{array}{ccll}
\toprule
\text{construction}&\text{trait slots}&p_i&p_j\\
\midrule
\text{split-opposed}&G&(13/20,9/10)&(11/20,1/10)\\
&H&(1/10,1/10,1/20)&(19/20,9/10,4/5)\\
\text{broad}&G&(1/20,9/10)&(19/20,11/20)\\
&H&(13/20,1/10,1/10)&(9/10,1/10,4/5)\\
\bottomrule
\end{array}
\]
Each row therefore has exactly the same core-probability multiset in the two
constructions, although the fixed $G/H$ slots receive different probability
pairs.

The shared-edge and outside probabilities are also unchanged. Conditional edge
independence therefore makes the complete joint law of
$(A_{ij},C_i,C_j,D_i,D_j)$ identical, where $C_r=X_r+Z_r$.

Finite convolution of the seven external Bernoulli variables in each row,
followed by the common shared-edge mixture, gives the degree-collision
probability
\[
 \Pp(D_i=D_j)=\frac{17057358033563}{10^{14}}>0.
\]
The pair-sign numerators are
\[
 -\frac{108361640879371}{8\cdot10^{14}}
 \quad\text{and}\quad
 \frac{7319767236631}{8\cdot10^{14}},
\]
so the corresponding collision-normalized targets are approximately
$-0.7940975$ and $0.0536408$.  Their signs are opposite despite the identical
collapsed experiment.

Finally, let the two focal factor vectors be the coordinate vectors $e_1,e_2$
and choose intercept $\mu=\logit(1/5)$.  For each alter $k$, set
\[
 u_k=
 \begin{pmatrix}
 \logit(p_{ik})-\mu\\
 \logit(p_{jk})-\mu
 \end{pmatrix}.
\]
Then $\mu+e_1^\top u_k=\logit(p_{ik})$,
$\mu+e_2^\top u_k=\logit(p_{jk})$, and
$\mu+e_1^\top e_2=\logit(1/5)$.  Thus each construction extends to a
symmetric rank-two logistic-Gram probability surface.
\end{proof}

\section{Growing Bivariate Nonredundancy}
\label{app:growing-witness}

\begin{proof}[Proof of Proposition~\ref{prop:growing-nonredundancy}]
Let $n=40m_n$ and repeat the two declared loading allocations at every common
midpoint-quantile activity site. Both allocations contain 14 positive and six
negative loadings in $B$, and six positive and 14 negative loadings in $O$.
They differ only in the allocation of the broad-role loadings between $G$ and
$H$.

Keep the $G/H/O$ trait slots fixed and reallocate the latent signs among the
$G/H$ slots. There is a population bijection between the two arrays that
preserves $B$ versus $O$, activity, and loading sign, although it need not
preserve the finer $G/H$ label. Match each broad/outside focal pair through
this bijection. After deleting the focal nodes, the matched arrays retain the
same multiset of $(\alpha_k,x_k)$ in $B$ and the same multiset in $O$. For
every remaining alter $k$, the joint pair of incident probabilities is
\[
\left(
\Lambda\!\left\{\mu+\alpha_i+\alpha_k+
\frac{c}{\sqrt n}x_ix_k\right\},
\Lambda\!\left\{\mu+\alpha_j+\alpha_k+
\frac{c}{\sqrt n}x_jx_k\right\}
\right).
\]
The multisets of these probability pairs agree separately in $B$ and $O$;
the shared focal-dyad probability also agrees. Multiplication of the Bernoulli
probability-generating functions proves equality of the full collapsed law in
the proposition for every matched focal pair and every $n$. Summing the equal
collision probabilities gives exact equality of the two target denominators;
Lemma~\ref{lem:collision-mass} makes their common order $\Theta(n)$.

The retained $G/H$ allocation changes the conditional bivariate response.
For a standard Gaussian pair $(N_1,N_2)$ of correlation $\rho$, write
\[
\mathfrak g_\rho(s,t)
=\E\{\sgn(N_1+s)\sgn(N_2+t)\}.
\]
Both arrays have opposite- and same-sign focal-pair weights $0.58$ and $0.42$.
Lemma~\ref{lem:marked-conditional}, uniformly over collision midpoints and
focal states, gives the broad response
\[
R_b(6,a)
=0.58\mathfrak g_{-1/3}\!\left(
2.4\sqrt{2v(a)/3},2.4\sqrt{2v(a)/3}\right)
+0.42\mathfrak g_{-1/3}(0,0)>0
\]
by Assumption~\ref{ass:response-margin}. For the split allocation, with
$q(a)=\sqrt{2v(a)/3}$,
\[
R_s(6,a)
=0.58\mathfrak g_{-1/3}\{6q(a),-1.2q(a)\}
+0.42\mathfrak g_{-1/3}(0,0)<0.
\]
The second inequality follows because a negatively correlated Gaussian pair
is negatively quadrant dependent: the two sign maps are increasing while
their marginal means under opposite shifts have opposite signs. Compactness
and $\inf_a v(a)>0$ make both margins uniform.

The exact plus/minus generating-polynomial identity makes the leading
degree-collision kernel common across focal sign types. The two-constraint
bound in Lemma~\ref{lem:collision-mass} also gives, in either design,
\[
\E(M_n^{\BO}-S_{2,n}^{\BO})=O(\sqrt n)=o(\E M_n^{\BO}).
\]
Thus collision normalization and nonzero-mark normalization have the same
signed limits. The two nonadditive arrays have the same collapsed law and
opposite bivariate responses, establishing nonredundancy of the separate group
queries.
\end{proof}

\section{Activity-Profiled Local Response}
\label{app:local-response}

Use a deterministic triangular array
$(\alpha_{i,n},x_{i,n},g_{i,n})$, with $g_{i,n}\in\{G,H,O\}$,
activities in a common compact interval $I$, and $|x_{i,n}|\le L$.
For $A\in\{G,H,O\}$ and $r\in\{0,1,2\}$, define
\[
F_{A,r,n}(a)
=\sum_{i:g_{i,n}=A}x_{i,n}^r\ind\{\alpha_{i,n}\le a\}.
\]

\begin{assumption}[Deterministic marked profile]
\label{ass:marked-profile}
There are bounded piecewise-Lipschitz functions $\nu_{A,r}$, with a common
finite breakpoint set, such that
\[
\Delta_n
=\max_{A,r}\sup_{a\in I}
\left|F_{A,r,n}(a)
-n\int_{I_-}^{a}\nu_{A,r}(t)\,dt\right|
=o(\sqrt n).
\]
Write $\lambda_A=\nu_{A,0}$. Each role has positive limiting share,
$\lambda_A$ is bounded away from zero on the interior of every declared
positive-support interval, and
\[
\nu_{A,2}(a)\lambda_A(a)-\nu_{A,1}(a)^2\ge0
\]
at continuity points with $\lambda_A(a)>0$. Moreover, for every interval
$J\subset I$,
\[
N_{A,n}(J)
=\#\{i:g_{i,n}=A,\ \alpha_{i,n}\in J\}
\le C\{1+n|J|\}.
\]
For the local-response result, impose the common-activity specialization
$\lambda_A(a)=\pi_A f(a)$ for a common positive density $f$.
\end{assumption}

The repeated 40-node midpoint-quantile arrays used in the main text satisfy
Assumption~\ref{ass:marked-profile} with
\[
\nu_{A,r}(a)=\frac{c_{A,r}}{40}f(a),
\]
where $c_{A,r}$ is the within-cell sum of $x^r$ over role $A$; their cumulative
discrepancy is $O(1)$.

At regular points with $\lambda_A(a)>0$, define the deterministic local-moment
ratios
\[
m_A(a)=\frac{\nu_{A,1}(a)}{\lambda_A(a)},
\qquad
s_A(a)=\frac{\nu_{A,2}(a)}{\lambda_A(a)}.
\]
These local moments are determined by the specified node array. Fix a collision midpoint
$\bar a$. For an alter of activity $a$, define
\[
p_{\bar a}(a)=\Lambda(\mu+\bar a+a),
\qquad
v_{\bar a}(a)=p_{\bar a}(a)\{1-p_{\bar a}(a)\}.
\]
\[
V_A=\int_I v_{\bar a}(a)\lambda_A(a)\,da,
\qquad
M_A=\int_I v_{\bar a}(a)\nu_{A,1}(a)\,da,
\]
and
\[
\widetilde M_A=M_A-\frac{V_A}{V_T}M_T,
\qquad
V_T=\sum_A V_A,
\quad M_T=\sum_A M_A.
\]
This is the information projection induced by conditioning on the total count;
under the common-activity specialization it removes a loading profile common
to all three groups. For the queried coordinates define
\[
\Sigma
=2\left[
\operatorname{diag}(V_G,V_H)
-\frac1{V_T}
\begin{pmatrix}V_G\\V_H\end{pmatrix}
\begin{pmatrix}V_G&V_H\end{pmatrix}
\right],
\]
\[
b_A=\frac{\widetilde M_A}{\sqrt{\Sigma_{AA}}},
\qquad
\varrho=\frac{\Sigma_{GH}}
{\sqrt{\Sigma_{GG}\Sigma_{HH}}}.
\]
Let $\lambda_B=\lambda_G+\lambda_H$ and define the broad-group focal moments
\[
m_B(a)=\frac{\nu_{G,1}(a)+\nu_{H,1}(a)}{\lambda_B(a)},
\qquad
s_B(a)=\frac{\nu_{G,2}(a)+\nu_{H,2}(a)}{\lambda_B(a)}.
\]
The collision-midpoint focal contrast is
\[
\chi_{BO}(\bar a)
=s_B(\bar a)+s_O(\bar a)-2m_B(\bar a)m_O(\bar a)
\]
for independently sampled focal roles in the natural $B/O$ frame.

Consider the rank-one local sequence
\[
\logit(p_{ij,c})
=\mu+\alpha_i+\alpha_j+\frac{c}{\sqrt n}x_ix_j.
\]
Let $R(c)$ denote the leading collision-conditioned Gaussian expectation of
the pair-sign kernel in the natural $B/O$ frame. Central symmetry of the
centered Gaussian limit makes every fixed focal-contrast response jointly even
in $c$, so $R'(0)=0$; no symmetry of $x_I-x_J$ is required.

\begin{proof}[Proof of Proposition~\ref{prop:curvature}]
Lemma~\ref{lem:marked-profile-summation} supplies the deterministic local
loading moments. Lemma~\ref{lem:marked-conditional} supplies the Gaussian
limit conditional on the degree collision. At one regular collision midpoint,
the three row-count differences have independent information masses
$2V_G,2V_H,2V_O$ before conditioning. Conditioning on their sum being zero takes the Schur
complement of the total-count direction.  The covariance of the queried
coordinates is therefore
\[
 \Sigma
 =2\left[
 \operatorname{diag}(V_G,V_H)
 -\frac1{V_T}
 \begin{pmatrix}V_G\\V_H\end{pmatrix}
 \begin{pmatrix}V_G&V_H\end{pmatrix}
 \right].
\]
The same conditioning projects the information-weighted loading masses to
\[
 \widetilde M_A=M_A-\frac{V_A}{V_T}M_T.
\]
Consequently, for one focal loading contrast $x_I-x_J$, the standardized
mean shifts in the two queried coordinates are proportional to
$b_G(x_I-x_J)$ and $b_H(x_I-x_J)$.

Let $(N_G,N_H)$ be standard bivariate normal with correlation $\varrho$.
Direct
differentiation at the origin of
\[
 (s,t)\longmapsto
 \E\{\sgn(N_G+s)\sgn(N_H+t)\}
\]
shows that the second directional derivative along $(b_G,b_H)$ is
\[
 \frac{2}{\pi\sqrt{1-\varrho^2}}
 \{2b_Gb_H-\varrho(b_G^2+b_H^2)\}.
\]
Central symmetry of $(N_G,N_H)$ makes the response for every fixed focal
contrast an even function of $c$, so the first derivative vanishes without a
symmetry assumption on the focal-contrast distribution. Averaging its square
at the collision midpoint gives the multiplier $\chi_{BO}(\bar a)$ and hence
the stated $\kappa$.

If the same bounded loading function is added to every group, the
common-activity specialization makes its change in $(M_G,M_H,M_O)$ proportional to
$(V_G,V_H,V_O)$, so the projection removes it exactly.  For the declared
broad profile, $(m_G,m_H,m_O)=(0.4,0.4,-0.4)$; for the declared split profile,
$(m_G,m_H,m_O)=(1,-0.2,-0.4)$.  To see the signs directly, write
$q=\int_I v_{\bar a}(a)f(a)\,da>0$.  The $10{:}10{:}20$ shares give
\[
 (V_G,V_H,V_O)=(q/4,q/4,q/2),
 \qquad \varrho=-1/3,
 \qquad \chi_{BO}(\bar a)=58/25.
\]
Both declared profiles have $M_T=0$.  Substitution into the quadratic form
therefore gives
\[
 \kappa_{\rm broad}=\frac{928q}{1875\pi\sqrt2}>0,
 \qquad
 \kappa_{\rm split}=-\frac{116q}{1875\pi\sqrt2}<0.
\]
These exact signs hold at every regular collision midpoint under the
common-activity specialization in Assumption~\ref{ass:marked-profile}.
Continuity of the curvature gives open sign-preserving classes around these
profiles. The expansion is taken at $c=0$; the broad-profile response at $c=6$ is evaluated in
Lemma~\ref{lem:positive-response}.
\end{proof}

\begin{proof}[Proof of Remark~\ref{rem:curvature-information}]
Let $u=(b_G+b_H)/\sqrt2$ and $v=(b_G-b_H)/\sqrt2$. Then
\[
I_{\mathrm{mean}}
=\frac{u^2}{1+\varrho}+\frac{v^2}{1-\varrho},
\]
and
\[
2b_Gb_H-\varrho(b_G^2+b_H^2)
=(1-\varrho^2)
\left\{\frac{u^2}{1+\varrho}-\frac{v^2}{1-\varrho}\right\}.
\]
Both denominators are positive. Taking absolute values and substituting into
Proposition~\ref{prop:curvature} proves the bound, since
$\chi_{BO}(\bar a)\ge0$. Equality holds on either mode axis; the sum axis
has nonnegative curvature and the difference axis has nonpositive curvature.
\end{proof}

\section{Primitive Open-Class Response and Power}
\label{app:full-pair-power}

\begin{lemma}[Positive natural-frame response at $c=6$]
\label{lem:positive-response}
Under the declared broad design,
\[
 \Pp(x=+1\mid B)=0.7,
 \qquad
 \Pp(x=+1\mid O)=0.3.
\]
The latent decomposition of every observable $B/O$ pair therefore has limiting
opposite-sign and same-sign weights $0.58$ and $0.42$.  If
\[
 v(a)=\lim_n\frac1n\sum_l\Lambda'(\mu+a+\alpha_l),
\]
then its collision-conditioned response obeys
\[
 R_{\rm broad}^{\BO}(c,a)
 =0.58g_{-1/3}\{0.4c\sqrt{2v(a)/3}\}
  +0.42g_{-1/3}(0).
\]
Under Assumption~\ref{ass:response-margin},
\[
 \inf_aR_{\rm broad}^{\BO}(6,a)>0.
\]
\end{lemma}

\begin{proof}
The opposite-sign weight is
$0.7(1-0.3)+(1-0.7)0.3=0.58$; the remaining weight is $0.42$.  For same-sign
focal pairs, the residual exponential tilt cancels from the row difference,
so the response tends to $g_{-1/3}(0)$.  For opposite-sign focal pairs, exact
midpoint factorization cancels the additive activity gap on the equal-degree
slice.  The broad loading profile then gives the common standardized shift in
the display.  The cellwise $20/20$ loading balance makes the leading collision
mass common across focal sign states, and Lemma~\ref{lem:marked-conditional}
transfers the Gaussian mixture uniformly to the lattice array. The explicit margin
assumption gives strict positivity at $c=6$. The mixture decomposes the
response by latent sign; the analyst uses every declared $B/O$ pair and never observes
the latent signs.
\end{proof}

\subsection{Group-specific marked kernels}

For $\vartheta\in\Theta_\varepsilon$, a local pair with slot types
$r\in B$, $s\in O$, activity anchor $a$, and scaled gap
$t=\sqrt n(\alpha_s-\alpha_r)$ has loading contrast
$d_{rs}=x_r-x_s$. Define
\[
Q_g(\vartheta,a)
=\frac1L\sum_{u:g(u)=g}
\int_0^1\Lambda'\{\mu+a+m_g(v)\}\,dv,
\]
\[
W_g(\vartheta,a)
=\frac1L\sum_{u:g(u)=g}x_u
\int_0^1\Lambda'\{\mu+a+m_g(v)\}\,dv,
\]
and write $Q=\sum_gQ_g$, $W=\sum_gW_g$.

\begin{lemma}[Primitive marked kernels]
\label{lem:primitive-kernels}
On a sufficiently small closed primitive neighborhood of $\vartheta_0$, the
degree constraint is recentered by
\[
\xi_{n,\vartheta}
=-\frac{c\,d_{rs}W}{2Q\sqrt n}+O(n^{-1}).
\]
After this correction, the limiting conditional mark mean and covariance are
\[
m_{GH}(\vartheta,a;r,s)
=c\,d_{rs}\left[
(W_G,W_H)-\frac WQ(Q_G,Q_H)
\right],
\]
\[
\Gamma(\vartheta,a)
=2\left[
\operatorname{diag}(Q_G,Q_H)
-\frac{(Q_G,Q_H)^\top(Q_G,Q_H)}Q
\right].
\]
If $\Psi(m,\Gamma)$ denotes the Gaussian sign-product expectation for a
bivariate normal vector with mean $m$ and covariance $\Gamma$, then
\[
J_\vartheta(r,s,a,t)
=\frac1{\sqrt{4\pi Q}}
\exp\left\{-\frac{(c\,d_{rs}W-tQ)^2}{4Q}\right\},
\]
\[
H_\vartheta(r,s,a,t)
=J_\vartheta(r,s,a,t)
\Psi\{m_{GH}(\vartheta,a;r,s),\Gamma(\vartheta,a)\}.
\]
The two kernels are jointly continuous in
$(\vartheta,r,s,a,t)$, with moving support endpoints interpreted through the
slot intensities, and satisfy
\[
|H_\vartheta(r,s,a,t)|+J_\vartheta(r,s,a,t)
\le Ce^{-c_0t^2}.
\]
\end{lemma}

\begin{proof}
Under group-specific profiles, the unweighted gauge
$\one^\top x=0$ need not imply $W=0$. The external residual tilt for one
increment is $cd_{rs}x_\ell/(2\sqrt n)$, so its first-order contribution to
the degree constraint is $cd_{rs}W\sqrt n/2$. The constraint Hessian is
$Qn+O(1)$, which gives the displayed information-weighted correction. The
remaining $O(n^{-1})$ adjustment restores the bounded focal-edge and lattice
offset exactly.

The normalized covariance of the degree constraint and the two group marks is
\[
2\begin{pmatrix}
Q&Q_G&Q_H\\Q_G&Q_G&0\\Q_H&0&Q_H
\end{pmatrix}.
\]
Its determinant is $8Q_GQ_HQ_O$, uniformly positive on the compact
neighborhood, and conditioning gives $\Gamma$. The same $O$, $G$, and $H$
unimodular anchors as in Lemma~\ref{lem:full-pair-kernel} control the central
and minor Fourier arcs. A bounded real mark tilt of order $n^{-1/2}$ changes
the constraint mean by $O(\sqrt n)$ and is recentered by another
$O(n^{-1/2})$ constraint correction without losing the anchors or covariance
bounds. Lemmas~\ref{lem:uniform-lattice}--\ref{lem:marked-conditional} then
give the two kernel formulas, including moment convergence.

Focal deletion and internal-edge conditioning alter means and covariance by
$O(1)$. The fixed $C^{1,1}$ envelope gives uniform Riemann-sum convergence and
$L^1$ continuity of the zero-extended slot intensities at moving endpoints.
The local change-of-measure affinity supplies the Gaussian envelope, so
dominated convergence gives joint kernel continuity.
\end{proof}

\begin{lemma}[Primitive positive-margin neighborhood]
\label{lem:primitive-margin}
There exist $\varepsilon_0>0$ and $\gamma_0>0$ such that
\[
\inf_{\vartheta\in\Theta_{\varepsilon_0}}
\frac{\mathcal M(\vartheta)}{\mathcal C(\vartheta)}
\ge\gamma_0>0,
\]
and $\inf_{\vartheta\in\Theta_{\varepsilon_0}}\mathcal C(\vartheta)>0$.
\end{lemma}

\begin{proof}
At $\vartheta_0$, the exact plus/minus generating-polynomial identity makes
the leading collision kernel common across focal sign types. Lemma
\ref{lem:positive-response} and
Assumption~\ref{ass:response-margin} therefore imply
\[
\mathcal M(\vartheta_0)/\mathcal C(\vartheta_0)>0.
\]
The $0.58/0.42$ mixture is used only at $\vartheta_0$.
Lemma~\ref{lem:primitive-kernels}, the Gaussian envelope, the
finite slot family, and $L^1$ continuity of the moving-support intensities
make $\mathcal M$ and $\mathcal C$ continuous in the primitive $C^1$ topology.
Positive broad/outside overlap makes $\mathcal C(\vartheta_0)>0$ and preserves
a positive lower bound on a sufficiently small closed neighborhood.
Continuity supplies the stated smaller relative-open neighborhood in the
primitive parameter coordinates and its positive response margin.
\end{proof}

\subsection{Uniform mean, variance, and clipped power}

\begin{lemma}[Open-class numerator bounds]
\label{lem:open-numerator}
Under either joint respondent-design/network law,
\[
\inf_{\vartheta\in\Theta_{\varepsilon_0}}
\frac1n\E_\vartheta U_n^{\mathsf S_n}\ge\kappa_0>0
\]
for all sufficiently large $n$, and
\[
\sup_{\vartheta\in\Theta_{\varepsilon_0}}
\Var_\vartheta(U_n^{\mathsf S_n})=O(n).
\]
\end{lemma}

\begin{proof}
The fixed-cell Riemann sum and Lemma~\ref{lem:primitive-kernels} give
\[
\E_\vartheta U_n^{\mathsf S_n}
=\rho_{2,n}\{\mathcal M(\vartheta)+o(1)\}n
\]
uniformly on the primitive ball, where
$\rho_{2,n}=\rho_n^2$ for Bernoulli respondents and
$\rho_{2,n}=(r_n)_2/(n)_2$ for fixed-size SRS. Lemma
\ref{lem:primitive-margin} gives the mean lower bound.

For the variance, same-pair terms sum to $O(n)$ by the one-constraint
collision envelope. For pairs sharing one respondent, the two-constraint
bound
\[
\Pp_\vartheta(D_i-D_j=d_1,D_i-D_k=d_2)
\le Cn^{-1}e^{-c_0n\{(\alpha_i-\alpha_j)^2+
(\alpha_i-\alpha_k)^2\}}
\]
and two local packing sums give $O(n)$. For four distinct respondents,
conditioning on the finite cross-edge state and applying the lattice
first-difference bound gives
\[
|m_e(b)-m_e(b')|
\le Cn^{-1}e^{-c_0n(\alpha_i-\alpha_j)^2}.
\]
The range-covariance inequality and two Gaussian packing sums again give
$O(n)$. Compact local tilts and the unimodular anchors make these bounds
uniform over the primitive neighborhood.

Finally, for respondent-pair masks $S_e$,
\[
\Cov(S_eh_e,S_fh_f)
=\pi_{e\cup f}\Cov(h_e,h_f)
+(\pi_{e\cup f}-\pi_e\pi_f)\E h_e\E h_f.
\]
Disjoint Bernoulli inclusions factor; under fixed-size SRS their covariance is
$O(n^{-1})$. The one-pair envelope gives
$\sup_i\sum_j|\E h_{ij}|=O(1)$, so the added shared and disjoint mask terms
also sum to $O(n)$.
\end{proof}

\begin{proof}[Proof of Theorem~\ref{thm:full-pair-open-power}]
The response statement is Lemma~\ref{lem:primitive-margin}. By
Lemma~\ref{lem:open-numerator} and Chebyshev's inequality,
\[
\inf_{\vartheta\in\Theta_{\varepsilon_0}}
\Pp_\vartheta\{U_n^{\mathsf S_n}\ge\kappa_0n/2\}\longrightarrow1.
\]
The clipped denominator is deterministically at most $nb_n$. On this event,
\[
\mathcal T_{n,\dagger}^S
\ge\frac{\kappa_0}{2}\sqrt{\frac n{b_n}}
\longrightarrow\infty
\]
because $b_n=\log(n+e)=o(n)$. Thus the uniform mean and variance bounds,
together with the deterministic clipping cap, establish consistency under
the joint respondent-design/network law.
\end{proof}

\begin{proof}[Proof of Corollary~\ref{cor:thinning-power}]
Lemma~\ref{lem:positive-response} and the uniform local-limit transfer give a
constant $\gamma_6>0$ such that
\[
 \liminf_n
 \frac{\E_6K_{R,n}^{\BO}}{\E_6S_{2,R,n}^{\BO}}
 \ge\gamma_6.
\]
Lemma~\ref{lem:collision-mass} gives
\[
\E(M_n^{\BO}-S_{2,n}^{\BO})=O(\sqrt n),
\]
while both denominators are order $n$. Equal-probability respondent sampling
contributes the same inclusion factor to every eligible pair. Hence
\[
 \frac{\E K_{R,n}^{\BO}}{\E S_{2,R,n}^{\BO}}
 =\frac{\E K_n^{\BO}}{\E S_{2,n}^{\BO}},
 \qquad
 \frac{\E S_{2,n}^{\BO}}{\E M_n^{\BO}}\longrightarrow1,
\]
which proves the target statement after reducing $\gamma_6$ if necessary.

Lemmas~\ref{lem:collision-mass}, \ref{lem:network-variance}, and
\ref{lem:respondent-transfer} give an order-$n$ nonzero-mark mass,
$S_{2,R,n}^{\BO}/n$ bounded away from zero, and
$\Var_6(K_{R,n}^{\BO})=O(n)$.  Lemma~\ref{lem:mask-clt} supplies the
conditional mask normal approximation.  The standardized positive center is
\[
 \frac{\pi_n\E_6K_{R,n}^{\BO}}
 {\{\pi_n\E_6S_{2,R,n}^{\BO}\}^{1/2}}
 \asymp \sqrt{n\pi_n}=n^{1/4}\longrightarrow\infty,
\]
whereas the centered network fluctuation is $o_p(1)$ by the same variance
ratio used under the null.  The rejection probability therefore tends to
one.
\end{proof}

\bibliographystyle{apalike}
\bibliography{ARD}

\end{document}

%% file: pair_sign_minimal_t5_finite_evidence_table.tex
\begin{table}[t]
\centering
\caption{Finite-sample behavior of the full-pair and thinning procedures. Panel A reports ranges across Bernoulli and fixed-size-SRS respondent sampling. Panel B reports paired comparisons at the largest prespecified design.}
\label{tab:minimal-t5-finite}
\scriptsize
\setlength{\tabcolsep}{2.5pt}
\textbf{Panel A: Rejection progression (percent)}\par\smallskip
\begin{tabular}{rrcccccc}
\toprule
$n$ & Resp. & N0-C full & N0-G full & P1-C full & P1-G full & P1-C thin & P1-G thin \\
\midrule
80 & 0.35 & 0.7--2.0 & 0.7--2.3 & 12.7--12.7 & 6.7--9.7 & 4.7--5.7 & 7.3--7.7 \\
80 & 0.60 & 0.0--0.7 & 0.7--1.3 & 17.3--19.0 & 14.7--14.7 & 7.0--7.3 & 5.3--8.0 \\
160 & 0.35 & 1.0--1.7 & 0.3--0.7 & 14.0--16.3 & 10.0--13.7 & 7.3--8.3 & 7.7--8.0 \\
160 & 0.60 & 0.0--0.3 & 0.0--0.3 & 17.3--24.0 & 9.0--12.0 & 7.0--8.7 & 5.7--6.7 \\
320 & 0.35 & 0.3--1.3 & 0.0--0.3 & 14.3--17.7 & 12.3--16.7 & 7.3--9.3 & 6.3--8.0 \\
320 & 0.60 & 0.0--0.0 & 0.0--0.0 & 29.0--29.0 & 13.3--17.0 & 8.0--9.3 & 4.3--4.7 \\
640 & 0.35 & 0.0--0.0 & 0.3--0.3 & 19.3--22.3 & 14.0--15.3 & 4.3--5.3 & 5.3--6.7 \\
640 & 0.60 & 0.0--0.0 & 0.0--0.0 & 43.7--45.3 & 18.7--24.7 & 5.7--8.0 & 7.0--7.7 \\
\bottomrule
\end{tabular}
\par\medskip
\textbf{Panel B: Paired comparison at $n=640$ (percent or percentage points)}\par\smallskip
\begin{tabular}{llrccc}
\toprule
P1 design & Respondent sampling & Resp. & Full & Thin & Full $-$ thin [95\%] \\
\midrule
common & Bernoulli & 0.35 & 22.3 & 5.3 & +17.0 [+12.1, +21.9] \\
common & Bernoulli & 0.60 & 43.7 & 8.0 & +35.7 [+29.6, +41.8] \\
common & fixed SRS & 0.35 & 19.3 & 4.3 & +15.0 [+9.9, +20.1] \\
common & fixed SRS & 0.60 & 45.3 & 5.7 & +39.7 [+33.3, +46.1] \\
group-specific & Bernoulli & 0.35 & 14.0 & 5.3 & +8.7 [+4.5, +12.8] \\
group-specific & Bernoulli & 0.60 & 24.7 & 7.7 & +17.0 [+11.7, +22.3] \\
group-specific & fixed SRS & 0.35 & 15.3 & 6.7 & +8.7 [+4.0, +13.4] \\
group-specific & fixed SRS & 0.60 & 18.7 & 7.0 & +11.7 [+6.6, +16.8] \\
\bottomrule
\end{tabular}
\par\smallskip
\begin{minipage}{0.97\linewidth}
\footnotesize Notes: C and G denote the repeated-common and group-specific smooth-profile designs; the latter uses prespecified nonbinary loadings. Each cell uses 300 networks and a one-sided nominal level of 5\%. Panel A reports descriptive ranges across the two respondent-sampling schemes. Panel B uses paired network and respondent draws; its normal intervals are descriptive and unadjusted for multiple comparisons.
\end{minipage}
\end{table}

%% file: pair_sign_minimal_t5_availability_table.tex
\begin{table}[t]
\centering
\caption{Finite-sample behavior when broad and outside roles have separated activity supports. Entries are ranges across Bernoulli and fixed-size-SRS respondent sampling.}
\label{tab:minimal-t5-availability}
\scriptsize
\setlength{\tabcolsep}{5pt}
\begin{tabular}{rrccc}
\toprule
$n$ & Resp. & Full rejection (\%) & Lower floor (\%) & Nonzero pairs/$n$ \\
\midrule
80 & 0.35 & 1.0--3.3 & 91.3--93.0 & 0.0215--0.0228 \\
80 & 0.60 & 8.0--9.3 & 50.7--54.0 & 0.0692--0.0716 \\
160 & 0.35 & 0.0--0.3 & 93.7--96.3 & 0.0107--0.0109 \\
160 & 0.60 & 5.7--6.3 & 68.0--68.0 & 0.0352--0.0359 \\
320 & 0.35 & 0.0--0.3 & 99.0--99.7 & 0.0043--0.0045 \\
320 & 0.60 & 1.3--2.0 & 87.0--92.0 & 0.0124--0.0126 \\
640 & 0.35 & 0.0--0.0 & 100.0--100.0 & 0.0008--0.0010 \\
640 & 0.60 & 0.0--0.0 & 98.7--100.0 & 0.0025--0.0027 \\
\bottomrule
\end{tabular}
\par\smallskip
\begin{minipage}{0.92\linewidth}
\footnotesize Notes: Each cell uses 300 networks. Separated activity supports make the procedure unavailable despite retaining the broad loading proportions. Collision mass vanishes over the reported size sequence, activating the lower scale floor and driving full-pair rejection to zero.
\end{minipage}
\end{table}